\documentclass{lmcs} 
\pdfoutput=1

\usepackage{lastpage}
\lmcsdoi{15}{3}{6}
\lmcsheading{}{\pageref{LastPage}}{}{}%
{Jul.~04,~2018}{Jul.~29,~2019}{}

\pdfoutput=1

\usepackage{hyperref}
\usepackage{amsfonts,amsmath,amssymb,wrapfig}
\usepackage{latexsym}
\usepackage{xspace}
\usepackage{rotating}
\usepackage{mathtools}
\usepackage{enumitem}

\usepackage{subfigure}
\usepackage{wasysym}

\usepackage{tikz}
\usetikzlibrary{positioning,shapes,petri}
\usetikzlibrary{automata}
\usetikzlibrary{arrows,backgrounds,positioning,fit,calc}
\usetikzlibrary{matrix}

\usepackage{macros} 

\newcommand{\arkiv}[1]{}

\theoremstyle{plain}
\begin{document}

\title{Controlling a population}

\author[N. Bertrand]{Nathalie Bertrand}	
\address{Univ Rennes, Inria \& IRISA, France}	
\email{nathalie.bertrand@inria.fr}  

\author[M. Dewaskar]{Miheer Dewaskar}
\address{University of North Carolina at Chapel Hill, USA}
\email{miheer@live.unc.edu}

\author[B. Genest]{Blaise Genest}	
\address{Univ Rennes, CNRS, IRISA, France}	
\email{blaise.genest@irisa.fr}  

\author[H. Gimbert]{Hugo Gimbert}	
\address{CNRS \& LaBRI, France}	
\email{hugo.gimbert@labri.fr}

\author[A.A. Godbole]{Adwait Amit Godbole}
\address{IIT Bombay, India}
\email{godbole15@gmail.com}




\begin{abstract}
  \noindent
  We introduce a new setting where a population of agents, each
  modelled by a finite-state system, are controlled uniformly: the
  controller applies the same action to every agent. The framework is
  largely inspired by the control of a biological system, namely a
  population of yeasts, where the controller may only change the
  environment common to all cells. We study a synchronisation problem
  for such populations: no matter how individual agents react to the
  actions of the controller, the controller aims at driving all agents
  synchronously to a target state. The agents are naturally
  represented by a non-deterministic finite state automaton (NFA), the
  same for every agent, and the whole system is encoded as a 2-player
  game. The first player (\playerone) chooses actions, and the second
  player (\playertwo) resolves non-determinism for each agent. The
  game with $m$ agents is called the $m$-population game. This gives
  rise to a parameterized control problem (where control refers to 2
  player games), namely the \emph{population control problem}: can
  \playerone\ control the $m$-population game for all
  $m \in \mathbb{N}$ whatever \playertwo\ does?

  In this paper, we prove that the population control problem is
  decidable, and it is a \EXPTIME-complete problem.  As far as we
  know, this is one of the first results on the control of
    parameterized systems. 
	 Our algorithm, which is not based on
  cut-off techniques, produces winning strategies which are symbolic,
  that is, they do not need to count precisely how the population is
  spread between states. 
  The winning strategies produced by our algorithm are optimal with respect to the synchronisation time: the maximal number of
    steps before synchronisation of all agents in the target state is
    at most polynomial in the number of agents $m$, and exponential in
    the size of the NFA.  
  We also show that if there is no winning strategy,
  then there is a population size $\cutoff$ such that \playerone\ wins
  the $m$-population game if and only if $m\leq \cutoff$.
  Surprisingly, $\cutoff$ can be doubly exponential in the number of
  states of the NFA, with tight upper and lower bounds. 
\end{abstract}

\maketitle

\section{Introduction}
Finite-state controllers, implemented by software, find applications
in many different domains: telecommunication, aeronautics, etc. Many
theoretical studies from the model-checking community showed that, in
idealised settings, finite-state controllers suffice. 
Games are an elegant formalism to model control
problems~\cite{ArnoldWalukiewicz}:
players represent the controller and the system; the precise setting
(number of players, their abilities, and their observations) depends
on the context.

Recently, finite-state controllers have been used to control living
organisms, such as a population of yeasts~\cite{Batt}. In this
application, microscopy is used to monitor the fluorescence level of a
population of yeasts, reflecting the concentration of some molecule,
which differs from cell to cell. Finite-state systems can model a
discretisation of the population of yeasts~\cite{Batt,AGKV16}.  The
frequency and duration of injections of a sorbitol solution can be
controlled, being injected uniformly into a solution in which the
yeast population is immersed. However, the response of each cell to
the osmotic stress induced by sorbitol varies, influencing the
concentration of the fluorescent molecule.  The objective is to
control the population to drive it through a sequence of predetermined
fluorescence states.

Taking inspiration from this biological control problem, we propose in
this paper, an \emph{idealised} problem for the population of yeasts:
the (perfectly-informed) controller aims at leading synchronously all
agents to a given fluorescence state. We introduce the
\emph{$m$-population game}, where a population of $m$ identical agents
is controlled uniformly. Each agent is modeled as a nondeterministic
finite-state automaton (NFA), the same for each agent. The first
player, called \playerone, applies the same action, a letter from the
NFA alphabet, to every agent. Its opponent, called \playertwo, chooses
the reaction of each individual agent, that is their successor state
upon that action. These reactions can differ due to non-determinism.
The objective for \playerone\ is to gather all agents synchronously in
the target state, and \playertwo\ seeks the opposite objective.
Our idealised setting may not be entirely satisfactory, yet it
constitutes a first step towards more realistic formalisations of the
yeast population control problem.

Dealing with large populations {\em explicitly} is in general
intractable due to the state-space explosion problem.
We therefore consider the associated {\em symbolic parameterized control problem}, that requires to synchronise all agents, independently of
the population size. Interestingly, this population control problem
does not fit traditional game frameworks from the model-checking
community. While {\em parameterized verification} received recently
quite some attention (see the related work below), to the best of our
knowledge, our framework is among the first ones in {\em parameterized
  control}.
 

\medskip {\bf Our results.}  We first show that considering an
infinite population is not equivalent to the parameterized control
problem
: there are simple cases where \playerone\ cannot control an infinite
population but can control every finite population.  Solving the
$\infty$-population game reduces to checking a reachability property
on the support graph~\cite{Martyugin-tocs14}, which can be easily done
in \PSPACE.  On the other hand, solving the parameterized control
problem requires new proof techniques, data structures and algorithms.

We easily obtain that when the answer to the population control
problem is negative, there exists a population size $\cutoff$, called
the \emph{cut-off}, such that \playerone\ wins the $m$-population game
if and only if $m\leq \cutoff$. Surprisingly, we obtain a lower-bound
on the cut-off doubly exponential in the number of states of the
NFA. 
Exploiting this cut-off naively would thus yield an inefficient
algorithm of least doubly exponential time complexity.

Fortunately, developing new proof
techniques 
({\em not} based on cut-off), we manage to obtain a better complexity:
we prove the population control problem to be \EXPTIME-complete. As a
byproduct, we obtain a doubly exponential upper-bound for the cut-off,
matching the lower-bound. Our techniques are based on a reduction to a
parity game with exponentially many states and a polynomial number of
priorities. The constructed parity game, and associated winning
strategies, gives insight on the winning strategies of \playerone\ in
the $m$-population games, for all values of $m$. 
\playerone\ selects
actions based on a polynomial number of {\em transfer graphs}, 
describing the trajectory of agents before reaching a given state.
If \playerone\ wins this parity game then he can uniformly apply his
winning strategy to all $m$-population games, just keeping track of
these transfer graphs, independently of the exact count in each
state. If \playertwo\ wins the parity game then he also has a uniform
winning strategy in $m$-population games, for $m$ large enough, which
consists in splitting the agents evenly among all transitions of the
transfer graphs.

Last, we obtain that when the answer to the population control problem is positive, the controller built by our algorithm takes at most a polynomial number of steps to synchronize all agents in 
the winning state, where the polynomial is of order the number of agents power the number of states of the NFA.
We show that our algorithm is optimal, as there are 
systems which require at least this order of steps to synchronize all agents.

\medskip {\bf Related work.}  Parameterized verification of systems with many identical components started with the seminal work of German and Sistla in the early nineties~\cite{GS-jacm92}, and received recently quite some attention.  The decidability and complexity of these problems typically depend on the communication means, and on whether the system contains a leader (following a different template) as exposed in the recent survey~\cite{Esparza-stacs14}.  This framework has been extended to timed automata templates~\cite{AJ-tcs03,ADRST-formats11} and probabilistic systems with Markov decision processes templates~\cite{BF-fsttcs13,BFS-fossacs14}.  Another line of work considers population 
protocols~\cite{AADFP-podc04,EGLM-concur15,BEK18}. 
Close in spirit, are broadcast protocols~\cite{esparza-verif-99}, in which one action may move an arbitrary number of agents from one state to another. Our model can be modeled as a subclass of broadcast protocols, where broadcasts emissions are self loops at a unique state, and no other synchronisation allowed. The parameterized reachability question considered for broadcast protocols is trivial in our framework, while our parameterized control question would be undecidable for broadcast protocols.  In these different works, components interact directly, while in our work, the interaction is indirect via the common action of the controller. Further, the problems considered in related work are  verification questions, and do not tackle the difficult issue 
we address of synthesizing a controller for all instances of a parameterized system.

There are very few contributions pertaining to parameterized games
with more than one player.
The most related is \cite{Kouv}, which proves
decidability of control of mutual exclusion-like protocols
in the presence of an unbounded number of agents.
Another contribution in that domain is the one of broadcast networks of identical parity games~\cite{BFS-fossacs14}. However, the game is used to solve a verification (reachability) question rather than a parameterized control problem as in our case. Also the roles of the two players are quite different.

The winning condition we are considering is close to {\em
  synchronising words}. The original synchronising word problem asks
for the existence of a word $w$ and a state $q$ of a {\em
  deterministic} finite state automaton, such that no matter the
initial state $s$, reading $w$ from $s$ would lead to state $q$ (see
\cite{Volkov-lata08} for a survey). Lately, synchronising words have
been extended to NFAs~\cite{Martyugin-tocs14}. Compared to our
settings, the author assumes a possibly infinite population of agents.
The setting is thus not parameterized, and
a usual support arena suffices to obtain a \PSPACE
algorithm. Synchronisation for probabilistic models~\cite{DMS12,DMS14}
have also been considered: the population of agents is not finite nor
discrete, but rather continuous, represented as a distribution.  The
distribution evolves deterministically with the choice of the
controller (the probability mass is split according to the
probabilities of the transitions), while in our setting, each agent
moves nondeterministically.
In \cite{DMS12}, the controller needs to apply the same action
whatever the state the agents are in (similarly to our setting), and
then the existence of a controller is undecidable. In \cite{DMS14},
the controller can choose the action depending on the state each agent
is in (unlike our setting), and the existence of a controller reaching
uniformly a set of states is \PSPACE-complete.

Last, our parameterized control problem can be encoded as a 2-player game on VASS~\cite{BJK-icalp10}, with one counter per state of the NFA: the opponent gets to choose the population size (a counter value), and the move of each agent corresponds to decrementing a counter and incrementing another. Such a reduction yields a symmetrical game on VASS in which both players are allowed to modify the counter values, in order to check that the other player did not cheat. Symmetrical games on VASS are undecidable~\cite{BJK-icalp10}, and their asymmetric variant (only one player is allowed to change the counter values) are decidable in 2\EXPTIME~\cite{JLS-icalp15}, thus with higher complexity than our specific parameterized control problem.

An extended abstract of this work appeared in the
  proceedings of the conference CONCUR 2017. In comparison, we provide
  here full proofs of our results, and added more intuitions and
  examples to better explain the various concepts introduced in this
  paper. Regarding contributions, we augmented our results with the
  study of the maximal time to synchronisation, 
  for which we show that 
  the controller built by our algorithm is optimal.

\medskip {\bf Outline.} In Section~\ref{sec:setting} we define the
population control problem and announce our main
results. Section~\ref{sec:capacity-game} introduces the capacity game,
and shows its equivalence with the population control problem
problem. Section~\ref{sec:parity} details the resolution of the
capacity game in \EXPTIME, relying on a clever encoding into a parity
game. It also proves a doubly exponential bound on the
cut-off. Section~\ref{sec:timetosynch} studies the maximal time to
synchronisation. Section~\ref{sec:lowerbounds} provides matching lower
bounds on the complexity and on the cut-off. The paper ends with a
discussion in Section~\ref{sec:conclu}.


\section{The population control problem}
\label{sec:setting}
\subsection{The $m$-population game}
A nondeterministic finite automaton (NFA for short) is a tuple
$\nfa=(\states, \Sigma, \state_0,\Delta)$ with $\states$ a finite set
of states, $\Sigma$ a finite alphabet, $\state_0 \in Q$ an initial state, and $\Delta \subseteq Q\times \Sigma \times Q$ the transition
relation.  
We assume throughout the paper that NFAs are complete, that
is, $\forall \state \in \states, a \in \Sigma \;, \exists p \in Q:
\;(\state,a,p) \in \Delta$. In the following, incomplete NFAs,
especially in figures, have to be understood as completed with a sink
state.

For every integer $m$, we consider a system $\nfa^{m}$ with $m$
identical agents $\nfa_1, \ldots, \nfa_{m}$ of the NFA $\nfa$.
The system $\nfa^{m}$ is itself an NFA $(\states^m, \Sigma,
\state^{m}_0,\Delta^m)$ defined as follows.  
Formally, states of
$\nfa^{m}$ are called configurations, and they are tuples $\vstate =
(\state_1,\ldots,\state_{m})\in Q^{m}$ describing the current state of each agent in the population.  We use the shorthand $\vstateinit[m]$,
or simply $\vstateinit$ when $m$ is clear from context, to denote the
initial configuration $(\stateinit,\ldots,\stateinit)$ of $\nfa^{m}$.
Given a target state $\targetstate \in Q$, the
$\targetstate$-synchronizing configuration is $\targetstate^{m} =
(\targetstate,\ldots,\targetstate)$ in which each agent is in the target state.

The intuitive semantics of $\nfa^m$ is that at each step, the same
action from $\Sigma$ applies to all agents. The effect of the action
however may not be uniform given the nondeterminism present in $\nfa$:
we have $((\state_1,\ldots,\state_m),a,(\state'_1,\ldots,\state'_m))
\in \Delta^m$ iff $(\state_j,\action,\state'_j)\in \Delta$ for all
$j \leq m$. A (finite or infinite) play in $\nfa^m$ is an alternating
sequence of configurations and actions, starting in the initial
configuration: $\play =\vstateinit \action_0 \vstate_1 \action_1
\cdots$ such that $(\vstate_i,\action_i,\vstate_{i+1}) \in
\Delta^m$ for all $i$. 

This is the \emph{$m$-population game} between \playerone\ and
\playertwo, where \playerone\ chooses the actions and \playertwo\
chooses how to resolve non-determinism. The objective for \playerone\
is to gather all agents synchronously in $\targetstate$ while
\playertwo\ seeks the opposite objective.

Our parameterized control problem asks whether \playerone\ can win the
$m$-population game for every $m \in \nats$. A strategy of \playerone\
in the $m$-population game is a function mapping finite plays to
actions, $\strat: (\states^m \times \actions)^* \times \states^m \to
\Sigma$.  A play $\play =\vstateinit \action_0 \vstate_1 \action_1
\vstate_2\cdots$ is said to \emph{respect $\sigma$}, or is a
\emph{play under $\strat$}, if it satisfies $\action_i = \strat(
\vstateinit \action_0 \vstate_1 \cdots \vstate_i)$ for all $i\in
\nats$. A play $\play =\vstateinit \action_0 \vstate_1 \action_1
\vstate_2\cdots$ is \emph{winning} if it hits the
$\targetstate$-synchronizing configuration, that is $\vstate_j =
\targetstate^m$ for some $j\in \nats$.  \playerone\ wins the
$m$-population game if he has a strategy such that all plays under
this strategy are winning. One can assume without loss of generality that $f$ is a sink state. If not, it suffices to add a new action leading tokens from $f$ to the new target sink state $\smiley$
and tokens from other states to a losing sink state $\frownie$.
The goal of this paper is to study the following parameterized
control problem:

 \begin{fmpage}{0.99\textwidth}
\textsf{Population control problem}\\
  {\bf Input}: An NFA $\nfa = (\states,\stateinit,\actions,\trans)$ 
  and a target state $\targetstate \in \states$.\\
  {\bf Output}: Yes iff for every integer $m$ \playerone\ wins the $m$-population game. \end{fmpage}

	
		
		
		
		

\begin{figure}[t!]
\centering
\begin{tikzpicture}
\draw(-2,0) node [circle,draw,inner sep=2pt,minimum
size=12pt] (s1) {$\stateinit$};

\draw(0,1) node [circle,draw,inner sep=2pt,minimum size=12pt] (s2)
{$\state_1$};

\draw(0,-1) node [circle,draw,inner sep=2pt,minimum size=12pt] (s3)
{$\state_2$};

\draw(1.5,0) node [circle,draw,inner sep=2pt,minimum size=12pt] (s4) {$\targetstate$};

\draw [-latex'] (s1) -- (s2) node [pos=.5,below] {$\delta$};
\draw [-latex'] (s1) -- (s3) node [pos=.5,above] {$\delta$};

\draw [-latex']  (s2) .. controls +(60:30pt) and +(120:30pt) .. (s2)
node[midway,above]{$\delta$};
\draw [-latex']  (s3) .. controls +(240:30pt) and +(300:30pt) .. (s3)
node[midway,below]{$\delta$};

\draw [-latex'] (s2) .. controls +(160:1cm) and +(60:1cm)  .. (s1) node [pos=.5,above] {$b$};
\draw [-latex'] (s2) -- (s4) node [pos=.5,above] {\quad \, $a$};
\draw [-latex'] (s3) -- (s4) node [pos=.5,below] {\quad \, $b$};
\draw [-latex'] (s3) .. controls +(200:1cm) and +(300:1cm)  .. (s1) node [pos=.5,below] {$a$};

 \draw [-latex']  (s1) .. controls +(150:30pt) and +(210:30pt) .. (s1)
 node[midway,left]{$a,b$};
\draw [-latex']  (s4) .. controls +(30:30pt) and +(330:30pt) .. (s4)
node[midway,right]{$a,b,\delta$};
\end{tikzpicture}
	\caption{An example of NFA: The splitting gadget $\splitnfa$.}
		\label{fig:splitgadget}	
\end{figure}

\medskip

For a fixed $m$, the winner of the $m$-population game can be
determined by solving the underlying reachability game with $|Q|^m$
states, which is intractable for large values of $m$.  
On the other hand, the answer to
the population control problem gives the winner of the $m$-population
game for arbitrary large values of $m$.  To obtain a decision
procedure for this 
parameterized problem, 
new data structures and algorithmic tools 
need to be developed, much more elaborate 
than the standard algorithm solving reachability games.

\begin{exa}
\label{ex}
We illustrate the population control problem with the example
$\splitnfa$ on action alphabet $\actions = \{a,b,\delta\}$ in
Figure~\ref{fig:splitgadget}.  Here, to represent a configuration
$\vstate$, we use a counting abstraction, and identify $\vstate$ with
the vector $(n_0,n_1,n_2,n_3)$, where $n_0$ is the number of agents in
state $\state_0$, etc, and $n_3$ is the number of agents in
$\targetstate$.  \playerone\ has a way to gather all agents
synchronously to $\targetstate$. We can give a symbolic representation
of a memoryless winning strategy $\strat$:
$\forall k_0,k_1 >0,\ \forall k_2,k_3 \geq 0,\ \strat(k_0,0,0,k_3) =
\delta,\ \strat(0,k_1,k_2,k_3) = a,\ \strat(0,0,k_2,k_3)=b$. Under
this strategy indeed, the number of agents outside $\targetstate$
decreases by at least one at every other step. The properties of this
example will be detailed later and play a part in proving a lower
bound (see Proposition~\ref{prop:cutoff-lowerbound}).
\end{exa}

\begin{exa}
We provide another illustrating example,
requiring a more involved strategy. 
Consider the NFA from Figure~\ref{fig:tryretry}, with
  $\actions = \{{\tt try},{\tt retry},{\tt top},{\tt bot},{\tt
    keep},{\tt restart}\}$. This NFA is again a positive instance of
  the population control problem. Yet, in contrast with the previous
  example, there are unsafe moves for \playerone. Indeed, after
  playing {\tt try} from $\stateinit$, playing ${\tt bot}$ 
  is losing if there are agents in $q_\top$,
  and playing ${\tt top}$ 
  is losing if there are agents in $q_\bot$ 
  (recall that unspecified transitions lead to a sink losing
  state). However, alternating ${\tt try}$ and ${\tt keep}$ until
  either $q_\bot$ becomes empty - allowing to play {\tt top} or 
   $q_\top$ is empty - allowing to play {\tt bot}, and then
  ${\tt restart}$, yields a configuration with less agents in
  $\stateinit$, and at least one in $\targetstate$. Continuing in the
  same way provides a winning strategy for \playerone. This example
  will be used again in Section~\ref{sec:timetosynch}, regarding the
  worst-case time to synchronisation.

\begin{figure}[!ht]
  \begin{center}
  \begin{tikzpicture}[-latex',>=stealth',shorten >=1pt,auto,node
    distance=2.2cm]
    \tikzstyle{every state}=[draw=black,text=black, inner
    sep=2pt,minimum size=12pt]

  \node[state] (A)                    {$\stateinit$};
  \node[state]         (B) [right of=A,yshift=1cm] {$\state_{\top}$};
  \node[state]         (C) [right of=A,yshift=-1cm] {$\state_{\bot}$};
  \node[state]         (D) [left of=C,yshift=-.5cm]       {$k$};
  \node[state]         (E) [right of=A,xshift=2cm]       {$\targetstate$};

  \path (A) edge  node [below] {\tt try} (B)
            edge              node [above]  {\tt try} (C)
        (B) edge [bend right] node [above left] {\tt keep} (A)
            edge              node  {\tt top} (E)
        (C) edge              node {\tt keep} (D)
            edge              node [below right] {\tt bot} (E)
        (D) edge              node {\tt restart} (A)
            edge [loop below] node {$\actions \setminus \{{\tt restart}\}$} (D)
        (E) edge [loop right] node {$\actions$} (E);
            
\end{tikzpicture}
\end{center}

\caption{A second example of NFA for the population control problem:
  $\cA_{\text{time}}$.}
\label{fig:tryretry}
\end{figure}
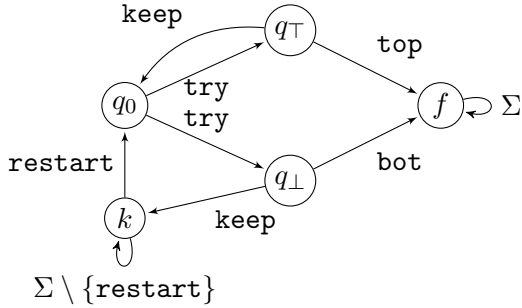
  \end{exa}

\subsection{Parameterized control and cut-off}
\label{sec:cutoff}
A first observation for the population control problem is that
$\vstateinit[m]$, $\targetstate^m$ and $Q^m$ are stable under a
permutation of coordinates.  A consequence is that the $m$-population
game is also symmetric under permutation, and thus the set of winning
configurations is symmetric and the winning strategy can be chosen
uniformly from symmetric winning configurations. Therefore, if
\playerone\ wins the $m$-population game then he has a positional
winning strategy which only counts the number of agents in each state
of $\nfa$ (the counting abstraction used in Example~\ref{ex}).


\begin{prop}
\label{prop:monotony}
  Let $m \in \nats$. 
  If \playerone\ wins the $m$-population game, then he wins the $m'$-population game for every $m' \leq m$.
 \end{prop}
\begin{proof}
  Let $m \in \nats$, and assume $\strat$ is a winning strategy for \playerone\ in $\nfa^m$. For $m' \leq m$ we define $\strat'$ as a strategy on $\nfa^{m'}$, inductively on the length of finite plays. Initially, $\strat'$ chooses the same first action as $\strat$: $\strat'(\stateinit^{m'}) = \strat(\stateinit^m)$. We then arbitrarily choose that the missing $m-m'$ 
agents would behave similarly as the first agent. This is indeed a possible move for the adversary in $\nfa^m$. Then, for any finite play under $\strat'$ in $\nfa^{m'}$, say $\play' = \vstateinit^{m'} \action_0 \vstate_1^{m'} \action_1 \vstate_2^{m'}\cdots \vstate_n^{m'}$, there must exist an extension $\play$ of $\play'$ obtained by adding $m-m'$ 
agents, all behaving as the first agent in $\nfa^{m'}$, that is consistent with $\strat$. Then, we let $\strat'(\play') = \strat(\play)$. Obviously, since $\strat$ is winning in $\nfa^m$, $\strat'$ is also winning in $\nfa^{m'}$.
\end{proof}


 Hence, when the answer to the population control problem is negative,
 there exists a \emph{cut-off}, that is a value $\cutoff \in \nats$
 such that for every $m < \cutoff$, \playerone\ has a winning strategy
 in $\nfa^m$, and for every $m \geq \cutoff$, he has no winning
 strategy.

\begin{exa}
  To illustrate the notion of cut-off, consider the NFA on alphabet
  $\actions = A \cup \{b\}$ from Figure~\ref{fig:linear_cutoff}. Here
  again, unspecified transitions lead to a sink losing state
  $\frownie$.

  Let us prove that the cut-off is $\cutoff = |Q|-2$ in this case. On
  the one hand, for $m < \cutoff$, there is a winning strategy
  $\strat_m$ in $\nfa^m$ to reach $\targetstate^m$, in just two
  steps. It first plays $b$, and because $m<\cutoff$, in the next
  configuration, there is at least one state $q_i$ such that no agent
  is in $q_i$. It then suffices to play $a_i$ to win. On the other
  hand, if $m \geq \cutoff$, there is no winning strategy to
  synchronize in $\targetstate$, since after the first $b$, agents can
  be spread so that there is at least one agent in each state
  $q_i$. From there, \playerone\ can either play action $b$ and
  restart the whole game, or play any action $a_i$, leading at least
  one agent to the sink state $\frownie$.
\end{exa}

\begin{figure}[htbp]
\centering
\begin{tikzpicture}
\draw(-2,0) node [circle,draw,inner sep=2pt,minimum
size=12pt] (s1) {$\stateinit$};

\draw(0,1) node [circle,draw,inner sep=2pt,minimum size=12pt] (s2)
{$\state_1$};
\draw(0,0) node [inner sep=2pt,minimum size=12pt] (s23) {$\vdots$};

\draw(0,-1) node [circle,draw,inner sep=2pt,minimum size=12pt] (s3)
{$\state_\cutoff$};

\draw(1.5,0) node [circle,draw,inner sep=2pt,minimum size=12pt] (s4) {$\targetstate$};

\draw [-latex'] (s1) -- (s2) node [pos=.5,above] {$b$};
\draw [-latex'] (s1) -- (s3) node [pos=.5,below] {$b$};

\draw [-latex'] (s2) .. controls +(160:1cm) and +(60:1cm)  .. (s1) node [pos=.5,above] {$b$};
\draw [-latex'] (s2) -- (s4) node [pos=.5,above] {\quad \, $A \setminus a_1$};
\draw [-latex'] (s3) -- (s4) node [pos=.5,below] {\quad \, $A \setminus a_\cutoff$};
\draw [-latex'] (s3) .. controls +(200:1cm) and +(300:1cm)  .. (s1) node [pos=.5,below] {$b$};

\draw [-latex']  (s4) .. controls +(30:30pt) and +(330:30pt) .. (s4)
node[midway,right]{$A \cup \{b\}$};
\end{tikzpicture}
	\caption{An NFA with a linear cut-off.}
		\label{fig:linear_cutoff}	
\end{figure}
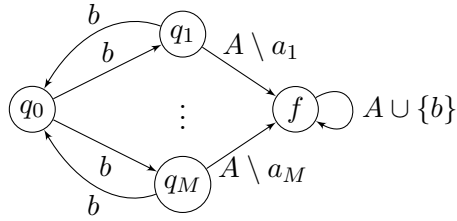


\subsection{Main results}

We are now in a position to state the contributions of this paper.
Most importantly, we establish 
the decidability 
and complexity 
of the population control problem, with matching upper and lower bounds on complexity:

\begin{thm}
\label{th1}
  The population control problem is \EXPTIME-complete.
\end{thm}

To prove Theorem~\ref{th1}, we proceed as follows. First,
Theorem~\ref{th.ror} states the equivalence of the population control
problem with an involved but non-parametric control problem, called
the capacity game. A simple yet suboptimal 2\EXPTIME\ upper bound
derives from this equivalence. In Theorem~\ref{th.parity}, we reduce
the capacity game to an exponential-size parity game with polynomially
many parities, yielding an \EXPTIME\ upper bound. The matching
\EXPTIME-hard lower bound is proved in Theorem~\ref{th.lower}.

\medskip For positive instances of the population control problem, our
decision algorithm computes a symbolic strategy $\sigma$, applicable
to all instances $\nfa^m$, which, in particular, does not rely on the
number of agents in each state. This symbolic strategy requires
exponential memory. Further, it is optimal with respect to the
synchronisation time, \emph{i.e.} the maximal number of steps before
synchronisation, which is polynomial in the number of agents.

\begin{thm}
  The synchronisation time under the winning strategy $\sigma$ is
  polynomial in the number of agents (and exponential in the size of
  $\nfa$). There is a family of NFA $(\nfa_n)$ with $n$ states, such
  that $m^{\frac{n-2}{2}}$ steps are needed by any strategy to
  synchronise $m$ agents.
\end{thm}

The upper bound is stated in Theorem~\ref{th.steps}, and the lower
bound in Corollary~\ref{c.exp}.

\medskip For negative instances to the population control problem, the
cut-off is at most doubly exponential, which is asymptotically tight.
\begin{thm}
\label{th2}
In case the answer to the population control problem is negative, the
cut-off is at most $\leq 2^{2^{O(|Q|^4)}}$. There is a family of
NFA $(\nfa_n)$ of size $O(n)$ and whose cut-off is $2^{2^{n}}$.
\end{thm}

Concerning the cut-off, the upper bound derives from results of
Theorem \ref{th.parity} (about the size of \playertwo' winning
strategy) combined with Proposition \ref{th.2wins}.  The lower bound
is stated in Proposition \ref{prop:cutoff-lowerbound}


\section{The \ror\ game}
\label{sec:capacity-game}
The objective of this section is to show that the population control
problem is equivalent to solving a game called the \emph{\ror\ game}.
To introduce useful notations, we first recall the population game
with infinitely many agents, as studied in \cite{Martyugin-tocs14} (see also \cite{Shi14} p.81).

\subsection{The $\infty$-population game}
\label{sec:resolution}
\label{subsec:game-support}
To study the $\infty$-population game, the behaviour of infinitely
many agents is abstracted into \emph{supports} which keep track of the
set of states in which at least one agent is. We thus introduce the
\emph{support game}, which relies on the notion of \emph{transfer
  graphs}. Formally, a transfer graph is a subset of $Q\times Q$
describing how agents are moved during one step.  The domain of a
transfer graph $G$ is $\dom(G) = \{ \state \in Q \mid \exists
(\state,r) \in G\}$ and its image is $\im(G) = \{r \in Q \mid \exists
(\state,r) \in G\}$.  Given an NFA $\nfa =(\states, \Sigma,
\state_0,\Delta)$ and $a \in \Sigma$, the transfer graph $G$ is
compatible with $a$ if for every edge $(q,r)$ of $G$, $(q,a,r) \in
\Delta$. We write $\GG$ for the set of transfer graphs.

The \emph{support game} of an NFA $\nfa$ is a two-player reachability
game played by \playerone\ and \playertwo\ on the \emph{support arena}
as follows. States are supports, \emph{i.e.}, non-empty subsets of
$\states$ and the play starts in $\{\state_0\}$. 
The goal support is $\{\targetstate\}$.
From a support $S$,
first \playerone\ chooses a letter $a \in \Sigma$, then \playertwo\
chooses a transfer graph $G$ compatible with $a$ and such that
$\dom(G)=S$, and the next support is $\im(G)$.  A play in the support
arena is described by the sequence $\rho = S_0 \arr{a_1,G_1} S_1
\arr{a_2,G_2}\ldots$ of supports and actions (letters and transfer
graphs) of the players.
Here, \playertwo' best strategy is to play the maximal graph possible, 
and we obtain a \PSPACE algorithm~\cite{Martyugin-tocs14}, and problem is
\PSPACE-complete:
\begin{prop}
\label{prop:support-implies-suresync}
\playerone\ wins the $\infty$-population game 
iff he wins the support game.
\end{prop}

\begin{proof}
  Let $\play =\vstateinit \action_1 \vstate_1
  \ldots \vstate_{n-1}\action_{n} \vstate_{n} \ldots$ be an infinite (or a finite) play of the $\infty$-population
  game:
agent $i \in \nats$ is in state $\vstate_k[i]$ at step $k$.  
  By only observing the support of the states and the transfer graphs, we
  can project this play onto the support arena. More precisely, 
  denoting $S_k = \{\vstate_k[i] \mid i \in \nats \}$ and $G_{k+1}=\{
  (\vstate_k[i],\vstate_{k+1}[i]) \mid  i \in \nats \}$
 for every $k$, 
  we have  $\Phi(\play) =
  S_0 \arr{\action_1,G_1} S_1 \cdots S_{n-1}
  \arr{\action_{n},G_n} S_{n} \cdots$ is a valid play in the \emph{support arena}. 
  
Hence if \playerone\ can win the support game with strategy $\sigma$, then
\playerone\ can use the strategy $\sigma \circ \Phi$ in the $\infty$-population game. This
is a winning strategy since the projection in the support arena should eventually reach $\{ f \}$.

On the other hand if \playerone\ doesn't have a winning strategy in the support
game, then by determinacy of reachability games, \playertwo\ has a strategy in
the support game to avoid reaching $\{f\}$. This strategy can be extended to a
strategy in the $\infty$-population game by sending infinitely many
agents along each edge of the chosen transfer graph. This can always be done
because, inductively, there are infinitely many agents in each state.
\end{proof}

Perhaps surprisingly, when it comes to finite populations, the support
game cannot be exploited to solve the population control
problem. Indeed, \playerone\ might win every $m$-population game (with
$m <\infty$) and at the same time lose the $\infty$-population game.
The example from Figure~\ref{fig:splitgadget} witnesses this
situation.  As already shown, \playerone\ wins any $m$-population game
with $m <\infty$. However, \playertwo\ can win the $\infty$-population
game by splitting agents from $q_0$ to both $q_1$ and $q_2$ each time
\playerone\ plays $\delta$.  This way, the sequence of supports is
$\{q_0\} \{q_1,q_2\} (\{q_0,f\} \{q_1,q_2,f\})^*$, which never hits
$\{f\}$.

\subsection{Realisable plays}

  
Plays of the $m$-population game (for $m <\infty$) can be abstracted
as plays in the support game, forgetting the identity of agents and
keeping only track of edges that are used by at least one agent.
Formally, given a play
$\play =\vstateinit \action_0 \vstate_1 \action_1 \vstate_2\cdots$ of
the $m$-population game, define for every integer $n$,
$S_n = \{\vstate_n[i] \mid 1\leq i\leq m\}$ and
$G_{n+1}=\{ (\vstate_n[i],\vstate_{n+1}[i]) \mid 1\leq i\leq m\}$.
We denote   $\Phi_m(\play)$ the play 
$S_0 \arr{\action_1,G_1} S_1 \arr{\action_2,G_2}\ldots$ 
in the support arena, called the projection of $\play$.

Not every play in the support arena can be obtained by projection.
This is the reason for introducing the notion of realisable plays:

\begin{defi}[Realisable plays]
A play of the support game is \emph{realisable} if there exists $m
<\infty$ such that it is the projection by $\Phi_m$ of a play
in the $m$-population game.
\end{defi}

To characterise realisability, we introduce entries of accumulators:


\begin{defi}
  Let $\rho = S_0 \arr{a_1,G_1} S_1 \arr{a_2,G_2}\ldots$ be a play in
  the support arena.  An \emph{accumulator} of $\rho$ is a sequence
  $T=(T_j)_{j\in \NN}$ such that for every integer $j$,
  $T_j \subseteq S_j$, and which is \emph{successor-closed}
  \emph{i.e.}, for every $j \in \nats$,
  $(s \in T_j \land (s,t)\in G_{j+1} )\implies t \in T_{j+1}\enspace.$
  For every $j \in \nats$, an edge $(s,t)\in G_{j+1}$ is an
  \emph{entry} to $T$ if $s\not \in T_j$ and $t\in T_{j+1}$; such an
  index $j$ is called an \emph{entry time}.
\end{defi}

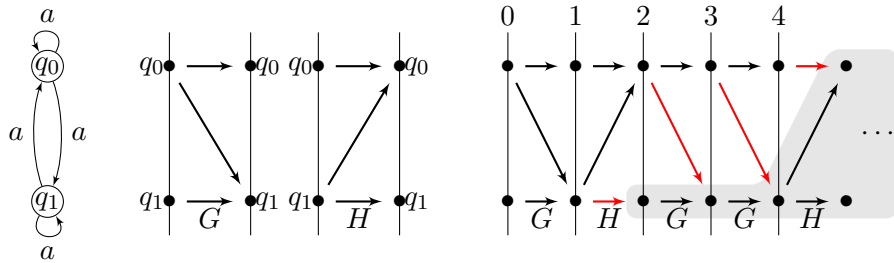
\begin{figure}[htbp]
\begin{center}
\begin{tikzpicture}[scale=0.9]
\draw (-1.8,0) node [circle,draw,inner sep=0pt,minimum size=12pt ]
  (s0) {$\state_0$} ;

  \draw(-1.8,-2) node [circle,draw,inner sep=0pt,minimum size=12pt ]
  (s1) {$\state_1$} ;

 \draw [-latex'] (s1) .. controls +(110:20pt) and +(250:20pt)  .. (s0)
 node [pos=.5,left] {$a$};
 \draw [-latex'] (s0) .. controls +(290:20pt) and +(70:20pt)  .. (s1)
 node [pos=.5,right] {$a$};

\draw [-latex']  (s0) .. controls +(60:20pt) and +(120:20pt) .. (s0)
node[midway,above]{$a$};

\draw [-latex']  (s1) .. controls +(240:20pt) and +(300:20pt) .. (s1)
node[midway,below]{$a$};


\draw (0,.5) -- (0,-2.5);
\draw (1.2,.5) -- (1.2,-2.5);

\draw (0,0) node (q0Gs) {$\bullet$};
\node  at (q0Gs.180) {$\state_0$}; 
\draw (0,-2) node (q1Gs) {$\bullet$};
\node  at (q1Gs.180) {$\state_1$}; 
\draw(1.2,0) node (q0Gt) {$\bullet$};
\node  at (q0Gt.0) {$\state_0$}; 
\draw(1.2,-2) node (q1Gt) {$\bullet$};
\node  at (q1Gt.0) {$\state_1$}; 

\draw[-latex',thick] (q0Gs) -- (q0Gt);
\draw[-latex',thick] (q0Gs) -- (q1Gt);
\draw[-latex',thick] (q1Gs) -- (q1Gt);

\node at (.6, -2.25) (G) {$G$};


\draw (2.2,.5) -- (2.2,-2.5);
\draw (3.4,.5) -- (3.4,-2.5);

\draw (2.2,0) node (q0Hs) {$\bullet$};
\node  at (q0Hs.180) {$\state_0$}; 
\draw (2.2,-2) node (q1Hs) {$\bullet$};
\node  at (q1Hs.180) {$\state_1$}; 
\draw(3.4,0) node (q0Ht) {$\bullet$};
\node  at (q0Ht.0) {$\state_0$}; 
\draw(3.4,-2) node (q1Ht) {$\bullet$};
\node  at (q1Ht.0) {$\state_1$};

\draw[-latex',thick] (q1Hs) -- (q0Ht);
\draw[-latex',thick] (q1Hs) -- (q1Ht);
\draw[-latex',thick] (q0Hs) -- (q0Ht);

\node at (2.8, -2.25) (H) {$H$};


\fill [fill=black!10,rounded corners] (6.75,-2.25) -- (10.75,-2.25)  -- (10.75,.25) --
(9.75,.25) -- (8.75,-1.75) -- (6.75,-1.75) -- cycle;

\node at (5,.75) {$0$};
\draw (5,.5) -- (5,-2.5);
\node at (6,.75) {$1$};
\draw (6,.5) -- (6,-2.5);

\node at (5.5,-2.25) (G1) {$G$};

\draw (5,0) node (q0Gs) {$\bullet$};
\draw (5,-2) node (q1Gs) {$\bullet$};
\draw(6,0) node (q0Gt) {$\bullet$};
\draw(6,-2) node (q1Gt) {$\bullet$};

\draw[-latex',thick] (q0Gs) -- (q0Gt);
\draw[-latex',thick] (q0Gs) -- (q1Gt);
\draw[-latex',thick] (q1Gs) -- (q1Gt);

\node at (7,.75) {$2$};
\draw (7,.5) -- (7,-2.5);

\draw (7,0) node (q0Hs) {$\bullet$};
\draw (7,-2) node (q1Hs) {$\bullet$};

\draw[-latex',thick] (q1Gt) -- (q0Hs);
\draw[-latex',thick,red] (q1Gt) -- (q1Hs);
\draw[-latex',thick] (q0Gt) -- (q0Hs);

\node at (6.5,-2.25) (H1) {$H$};
\node at (8,.75) {$3$};
\draw (8,.5) -- (8,-2.5);
\node at (9,.75) {$4$};
\draw (9,.5) -- (9,-2.5);
\draw (8,0) node (q0Gs) {$\bullet$};
\draw (8,-2) node (q1Gs) {$\bullet$};
\draw(9,0) node (q0Gi) {$\bullet$};
\draw(9,-2) node (q1Gi) {$\bullet$};

\draw[-latex',thick] (q0Hs) -- (q0Gs);
\draw[-latex',thick,red] (q0Hs) -- (q1Gs);
\draw[-latex',thick] (q1Hs) -- (q1Gs);
\draw[-latex',thick] (q0Gs) -- (q0Gi);
\draw[-latex',thick,red] (q0Gs) -- (q1Gi);
\draw[-latex',thick] (q1Gs) -- (q1Gi);

\node at (7.5,-2.25) (G2) {$G$};
\node at (8.5,-2.25) (G3) {$G$};

\draw(10,0) node (q0Hs) {$\bullet$};
\draw(10,-2) node (q1Hs) {$\bullet$};

\draw[-latex',thick,red] (q0Gi) -- (q0Hs);
\draw[-latex',thick] (q1Gi) -- (q0Hs);
\draw[-latex',thick] (q1Gi) -- (q1Hs);

\node at (9.5,-2.25) (G4) {$H$};

\node at (10.5,-1) (dots) {$\cdots$};
\end{tikzpicture}
\end{center}
\caption{An NFA, two transfer graphs, and a play with finite yet
  unbounded capacity.}
\label{fig:capacity}
\end{figure}

Figure~\ref{fig:capacity} illustrates the notions we just introduced:
it contains an NFA (left), two transfer graphs $G$ and $H$ (middle),
and $\rho = G H G^2 H G^3 \cdots$ a play in the support arena (right). The grey zone is an accumulator defined by
$T_0=T_1=\emptyset, T_2=T_3=T_4=\{q_1\}$ and $T_n = \{q_0,q_1\}$ for
all $n \geq 5$.

\begin{defi}[Plays with finite and bounded capacity]
A play has \emph{finite capacity} if all its accumulators have finitely many
entries (or entry times),
\emph{infinite capacity} otherwise, and  \emph{bounded capacity} if
the number of entries (or entry times) of its accumulators is bounded.
\end{defi}

Continuing with the example of Figure \ref{fig:capacity}, entries 
of the accumulator are depicted in red. 
The play $\rho = G H G^2 H G^3 \cdots$ is not realisable in any $m$-population game, since at
least $n$ agents are needed to realise $n$ transfer graphs $G$ in a
row: at each $G$ step, at least one agent moves from $q_0$ to $q_1$,
and no new agent enters $q_0$. Moreover, let us argue that $\rho$ has
unbounded capacity. A simple analysis shows that there are only two
kinds of non-trivial accumulators $(T_j)_{j \in \nats}$ depending on
whether their first non-empty $T_j$ is $\{\state_0\}$ or
$\{\state_1\}$. We call these top and bottom accumulators,
respectively. All accumulators have finitely many entries, thus the
play has finite capacity.  However, for every $n \in \nats$ there is a
bottom accumulator with $2 n$ entries.  Therefore, $\rho$ has
unbounded capacity, and it is not realisable.

\medskip

We show that in general, realisability is equivalent to {\em bounded}
capacity:

\begin{lem}\label{lem:noleak}
A play is realisable iff it has bounded capacity.
\end{lem}

\begin{proof}
  Let $\rho = S_0\arr{a_1,G_1} S_1 \arr{a_2,G_2}\cdots$ be a
  realisable play in the support arena and
  $\play=\vstate_0a_1\vstate_1a_2\vstate_2 \cdots$ a play in the
  $m$-population game for some $m$, such that $\Phi_m(\play) =
  \rho$. For any accumulator $T=(T_j)_{j\in \NN}$ of
  $\rho$, let us show that $T$ has less than $m$ entries.  For every
  $j\in\NN$, we define
  $n_j=\mid \{ 1 \leq k \leq m \mid \vstate_j(k) \in T_j\}\mid $ as
  the number of agents in the accumulator at index $j$.  By definition
  of the projection, every edge $(s,t)$ in $G_j$ corresponds to the
  move of at least one agent from state $s$ in $\vstate_j$ to state
  $t$ in $\vstate_{j{+}1}$.  Thus, since the accumulator is
  successor-closed, the sequence $(n_j)_{j\in\NN}$ is non-decreasing
  and it increases at each entry time. The number of entry times is
  thus bounded by $m$ the number of agents.

  Conversely, assume that a play
  $\rho = S_0\arr{a_1,G_1} S_1 \arr{a_2,G_2}\cdots$ has bounded
  capacity, and let $m$ be an upper bound on the number of entry times
  of its accumulators. Let us show that $\rho$ is the projection of a
  play $\play=\vstate_0a_1\vstate_1a_2\vstate_2 \cdots$ in the
  $(|S_0||Q|^{m+1})$-population game. In the initial configuration
  $\vstate_0$, every state in $S_0$ contains $|Q|^{m+1}$ agents.
  Then, configuration $\vstate_{n+1}$ is obtained from $\vstate_{n}$
  by spreading the agents evenly among all edges of $G_{n+1}$. As a
  consequence, for every edge $(s,t)\in G_{n+1}$ at least a fraction
  $\frac{1}{|Q|}$ of the agents in state $s$ in $\vstate_{n}$ moves to
  state $t$ in $\vstate_{n+1}$. By induction,
  $\play=\vstate_0 a_1 \vstate_1 a_2 \vstate_2 \cdots$ projects to
  some play $\rho' = S'_0\arr{a_1,G'_1} S'_1 \arr{a_2,G'_2}\cdots$
  such that for every $n\in \NN$, $S'_n\subseteq S_n$ and
  $G'_n\subseteq G_n$. To prove that $\rho'=\rho$, we show that for
  every $n\in \NN$ and state $t\in S_n$, at least $|Q|$ agents are in
  state $t$ in $\vstate_{n}$.  For that let
  $(U_j)_{j\in \{0\ldots n\}}$ be the sequence of subsets of $Q$
  defined by $U_n= \{t\}$, and for $0 < j < n$,
  \[
U_{j-1} = \{ s \in Q \mid \exists t' \in U_j, (s,t')\in G_j\}\enspace.
  \]

  In particular, $U_0=S_0$.  Let $(T_j)_{j\in \NN}$ be the sequence of
  subsets of states defined by $T_j= S_j \setminus U_j$ if $j\leq n$
  and $T_j=S_j$ otherwise.  Then $(T_j)_{j\in \NN}$ is an accumulator:
  if $s\not \in U_j$ and $(s,s')\in G_j$ then $s'\not \in U_{j{+}1}$.
  As a consequence, $(T_j)_{j\in \NN}$ has at most $m$ entry times,
  thus, there are at most $m$ indices $j\in \{0\ldots n-1\}$ such that
  some agents in the states of $S_{j} \setminus T_{j}=U_{j}$ in
  configuration $\vstate_{j}$ may move to states of $T_{j{+}1}$ in
  configuration $\vstate_{j{+}1}$.  In other words, if we denote $M_j$
  the number of agents in the states of $U_{j}$ in configuration
  $\vstate_{j}$ then there are at most $m$ indices where the sequence
  $(M_j)_{j\in \{0\ldots n\}}$ decreases.  By definition of $\play$,
  even when $M_j > M_{j{+}1}$, at least a fraction $\frac{1}{|Q|}$ of
  the agents moves from $U_{j}$ to $U_{j{+}1}$ along the edges of
  $G_{j{+}1}$, thus $M_{j{+}1} \geq \frac{M_j}{|Q|}$. Finally, the
  number of agents $M_n$ in state $t$ in $\vstate_n$ satisfies
  $M_n \geq \frac{|S_0||Q|^{m+1}}{|Q|^m} \geq |Q|$.  Hence $\rho$ and
  $\rho'$ coincide, so that $\rho$ is realisable.
\end{proof}


\subsection{The \ror\ game}
An idea to obtain a game on the support arena equivalent with
the population control problem is to 
make \playertwo\
lose whenever the play is not
realisable, \emph{i.e.} whenever the play has unbounded capacity.
One issue with (un)bounded capacity is however that it is not a regular property for runs. Hence, it is not easy to use it as a winning condition. On the contrary, \emph{finite} capacity is a regular property.
We thus relax (un)bounded capacity by using (in)finite capacity
and define the corresponding 
abstraction of the population game:

\begin{defi}[\Ror\ game]
  The \emph{\ror\ game} is the game played on the support arena, where
  \playerone\ wins a play iff either the play reaches
  $\{\targetstate\}$ or the play has infinite capacity.
  A player \emph{wins the \ror\ game} if he has a winning strategy in this game.
\end{defi}

We show that this relaxation can be used to decide the population control problem.

\begin{thm}
\label{th.ror}
The answer to the population control problem is positive iff
\playerone\ wins the \ror\ game.
\end{thm}

This theorem is a direct corollary of the following propositions
(\ref{th.determinacy} - \ref{th.2wins}):

\begin{prop}
  \label{th.determinacy}
  Either \playerone\  
  or \playertwo\ wins the capacity game, and
  the winner has a winning strategy with
  \emph{finite memory}.
\end{prop}
\begin{proof}
Whether a play has {infinite} capacity
can be verified by a non-deterministic B\"uchi automaton of size
$2^{|Q|}$ on the alphabet of transfer graphs, which guesses an
accumulator on the fly and checks that it has infinitely many entries.
This B\"uchi automaton can be determinised into a parity automaton
(\emph{e.g.}\ using Safra's construction) with state space $\MemSet$
of size $\mathcal{O}\left(2^{2^{|Q|}}\right)$.  The synchronized
product of this deterministic parity automaton with the support game
produces a parity game which is equivalent with the capacity game, in
the sense that, up to unambiguous synchronization with the
deterministic automaton, plays and strategies in both games are the
same and the synchronization preserves winning plays and strategies.
Since parity games are determined and positional~\cite{zielonka},
either \playerone\ or \playertwo\ has a positional winning strategy in
the parity game, thus either \playerone\ or \playertwo\ has a winning
strategy with finite memory $\MemSet$ in the capacity game.
\end{proof}

\begin{prop}
If \playerone\ wins the \ror\ game,
then \playerone\ has a winning strategy in the
$m$-population game for all $m$. 
\label{th.1wins}
\end{prop}
\begin{proof}
Assuming that \playerone\ wins the \ror\ game with a strategy $\sigma$,
he can win any
$m$-population game, $m <\infty$, with the strategy $\sigma_m=\sigma\circ \Phi_m$.
The projection $\Phi_m(\pi)$ of every infinite play $\pi$ respecting $\sigma_m$ is realisable,
thus $\Phi_m(\pi)$ has bounded, hence finite, capacity (Lemma~\ref{lem:noleak}).
Moreover $\Phi_m(\pi)$ respects $\sigma$,
and since $\sigma$ wins the \ror\ game, 
$\Phi_m(\pi)$ reaches $\{\targetstate\}$.
Thus $\pi$ reaches $\targetstate^{m}$ and $\sigma_m$ is winning.
\end{proof}

We now prove the more challenging reverse implication.  Recall by
Proposition \ref{th.determinacy} that if \playertwo\ has a winning
strategy in the \ror\ game, then he has a \emph{finite-memory}
strategy.

\begin{prop}
  \label{th.2wins}
  If \playertwo\ has a winning strategy in the \ror\ game
  using \emph{ finite memory} $M$,   then he has a winning strategy in the ${|\stnb|}^{ 1+ |\MemSet| \cdot 4^{|\stnb|}}$-population game.
\end{prop}

\begin{proof}
  Let $\tau$ be a winning strategy for \playertwo\ in the \ror\ game
  with finite-memory $\MemSet$. First we show that any play
  $\pi=S_0\arr{a_1,G_1} S_1 \arr{a_2,G_2}\ldots$ compatible with
  $\tau$ should have capacity (\emph{i.e.} count of entry times of any
  of its accumulator) bounded by $ B = |\MemSet| \times 4^{|\stnb|}$.

Let $ \{T_i\}_{i \in \NN}$ be any accumulator of $\pi$. If there are two integers $0\leq i < j \leq n$
  such that at times $i$ and $j$:
  \begin{itemize}
\item the memory state of $\tau$ coincide: $\mathsf{m}_i=\mathsf{m}_j$;
\item the supports coincide: $S_i=S_j$; and
\item the supports in the accumulator $T$ coincide: $T_i=T_j$.
\end{itemize}
then we show that there is no entry in the accumulator between indices
$i$ and $j$.  The play $\play_*$ identical to $\play$ up to date $i$
and which repeats ad infinitum the subplay of $\play$ between times
$i$ and $j$, is consistent with $\tau$, because
$\mathsf{m}_i=\mathsf{m}_j$ and $S_i=S_j$.  The corresponding sequence
of transfer graphs is $G_0,\ldots , G_{i-1} (G_i, \ldots, G_{j-1})
^\omega$, and $T_0,\ldots ,T_{i-1}(T_i\ldots T_{j-1})^\omega$ is a
``periodic" accumulator of $\play_*$.  By periodicity, this accumulator
has either no entry or infinitely many entries after date $i-1$.
Since $\tau$ is winning, $\play_*$ has finite capacity, thus the
periodic accumulator has no entry after date $i-1$, and hence there are no
entries in the accumulator $(T_j)_{j\in \NN}$ between indices $i$ and
$j$.

Let $I$ be the set of entry times for the accumulator
$(T_j)_{j\in \NN}$.  According to the above, for all pairs of distinct
indices $(i,j)$ in $I$, we have
$m_i\neq m_j\lor S_i\neq S_j \lor T_i\neq T_j$.  As a consequence,
\[
 |I| \leq B=|\MemSet|\cdot 4^{|\stnb|}\enspace.
\]

Now following the proof of Lemma \ref{lem:noleak}, for $m =
|\stnb|^{B + 1}$, \playertwo\ has a strategy
$\tau_m$ in the $m$-population game of following the transfer graphs suggested
by $\tau$. In other words, when it is \playertwo's turn to play in the
$m$-population game, and the play so far $\play=\vstate_0\arr{a_1}\vstate_1 \cdots \vstate_n
\arr{a_{n+1}}$ is projected via $\Phi_m$ to a play $\rho = S_0\arr{a_1,G_1} S_1 \cdots S_n \arr{a_{n+1}}$ in the capacity game,
let $G_{n+1}=\tau(\rho)$ be the decision of \playertwo\ at this point in the capacity game.  Then, to determine $\vstate_{n+1}$,
$\tau_m$ splits evenly the agents in $\vstate_n$ along every edge of $G_{n+1}$.
Since the capacity of $\rho$ is bounded by $B$, the argument in the proof of
Lemma \ref{lem:noleak} shows that $\vstate_n$ has at least $|Q|$ agents in each state
and thus $\{ (\vstate_n[i], \vstate_{n+1}[i]) \mid 1 \leq i \leq m \} =
G_{n+1}$. This means that the projected play $\Phi_m(\pi \vstate_{n+1})$ continues to be
consistent with $\tau$ and its support will never reach $\{ f \}$. Thus $\tau_m$
guarantees that not all agents will be in target state $m$ simultaneously, and
hence is a winning strategy for \playertwo\ in the $m$-population game.
\end{proof}

As consequence of Proposition~\ref{th.determinacy}, the population control
problem can be decided by explicitly computing the parity game and
solving it, in 2\EXPTIME.
In the next section we will improve this complexity bound to \EXPTIME.

\medskip

  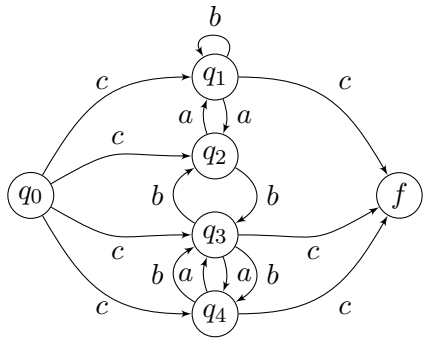
\begin{figure}[h!]
\centering
\begin{tikzpicture}[scale=0.7]
\draw (-3.5,.75) node [circle,draw,inner sep=2pt,minimum size=12pt ]
  (s0) {$\state_0$} ;
\draw (3.5,.75) node [circle,draw,inner sep=2pt,minimum size=12pt ]
  (sf) {$\targetstate$} ;

  \draw(0,3) node [circle,draw,inner sep=2pt,minimum size=12pt ]
  (s1) {$\state_1$} ;

\draw(0,1.5) node [circle,draw,inner sep=2pt,minimum size=12pt] (s2) {$\state_2$}
;

\draw(0,0) node [circle,draw,inner sep=2pt,minimum size=12pt] (s3) {$\state_3$};

\draw(0,-1.5) node [circle,draw,inner sep=2pt,minimum size=12pt] (s4){$\state_4$};

 \draw [-latex'] (s2) .. controls +(110:20pt) and +(250:20pt)  .. (s1)
 node [pos=.5,left] {$a$};
 \draw [-latex'] (s1) .. controls +(290:20pt) and +(70:20pt)  .. (s2)
 node [pos=.5,right] {$a$};

 \draw [-latex'] (s4) .. controls +(110:20pt) and +(250:20pt)  .. (s3)
 node [pos=.5,left] {$a$};
 \draw [-latex'] (s3) .. controls +(290:20pt) and +(70:20pt)  .. (s4)
 node [pos=.5,right] {$a$};

 \draw [-latex'] (s4) .. controls +(150:30pt) and +(210:30pt)  .. (s3)
 node [pos=.5,left] {$b$};
 \draw [-latex'] (s3) .. controls +(330:30pt) and +(30:30pt)  .. (s4)
 node [pos=.5,right] {$b$};

 \draw [-latex'] (s3) .. controls +(150:30pt) and +(210:30pt)  .. (s2)
 node [pos=.5,left] {$b$};
 \draw [-latex'] (s2) .. controls +(330:30pt) and +(30:30pt)  .. (s3)
 node [pos=.5,right] {$b$};

\draw [-latex']  (s1) .. controls +(60:30pt) and +(120:30pt) .. (s1)
node[midway,above]{$b$};

\draw [-latex'] (s0) .. controls +(60:2cm) and +(180:2cm)  .. (s1)
 node [pos=.5,above] {$c$};
\draw [-latex'] (s0) .. controls +(30:2cm) and +(180:2cm)  .. (s2)
 node [pos=.5,above] {$c$};
\draw [-latex'] (s0) .. controls +(330:2cm) and +(180:2cm)  .. (s3)
 node [pos=.5,below] {$c$};
\draw [-latex'] (s0) .. controls +(300:2cm) and +(180:2cm)  .. (s4)
 node [pos=.5,below] {$c$};

\draw [-latex'] (s1) .. controls +(0:2cm) and +(120:2cm)  .. (sf)
 node [pos=.5,above right] {$c$};
\draw [-latex'] (s3) .. controls +(0:2cm) and +(210:2cm)  .. (sf)
 node [pos=.5,below] {$c$};
\draw [-latex'] (s4) .. controls +(0:2cm) and +(240:2cm)  .. (sf)
 node [pos=.5,below right] {$c$};
\end{tikzpicture}
\caption{NFA where \playerone\ needs memory to win the associated \ror\ game.}
\label{fig:leak}
\end{figure}

We conclude with an example showing  that, in general, positional strategies are not sufficient to win the \ror\ game.
%
Consider the example of Figure~\ref{fig:leak}, where the only way for
\playerone\ to win is to reach a support without $\state_2$ and play
$c$.  With a memoryless strategy, \playerone\ cannot win the capacity
game.  There are only two memoryless strategies from support
$S=\{\state_1,\state_2,\state_3,\state_4\}$.  If \playerone\ only
plays $a$ from $S$, the support remains $S$ and the play has bounded
capacity.  If he only plays $b$'s from $S$, then \playertwo\ can split
agents from $\state_3$ to both $\state_2,\state_4$ and the play
remains in support $S$, with bounded capacity. In both cases, the play
has finite capacity
and \playerone\ loses.

However, \playerone\ can win the capacity game. His (finite-memory)
winning strategy $\strat$ consists in first playing $c$, and then
playing alternatively $a$ and $b$, until the support does not contain
$\{\state_2\}$, in which case he plays $c$ to win.  Two consecutive
steps $ab$ send $\state_2$ to $\state_1$, $\state_1$ to $\state_3$,
$\state_3$ to $\state_3$, and $\state_4$ to either $\state_4$ or
$\state_2$.  To prevent \playerone\ from playing $c$ and win,
\playertwo\ needs to spread from $\state_4$ to both $\state_4$ and
$\state_2$ every time $ab$ is played. Consider the accumulator $T$
defined by $T_{2i}=\{\state_1,\state_2,\state_3\}$ and
$T_{2i-1}=\{\state_1,\state_2,\state_4\}$ for every $i >0$. It has an
infinite number of entries (from $\state_4$ to $T_{2i}$).  Hence
\playerone\ wins if this play is executed.  Else, \playertwo\
eventually keeps all agents from $\state_4$ in $\state_4$ when $ab$ is
played, implying the next support does not contain
$\state_2$. Strategy $\strat$ is thus a winning strategy for
\playerone.

\section{Solving the \ror\ game in \EXPTIME}
\label{sec:parity}
To solve efficiently the \ror\ game, we build an equivalent
exponential size parity game with a polynomial number of parities.  To
do so, we enrich the support arena with a \emph{tracking list}
responsible of checking whether the play has finite capacity.  The
tracking list is a list of transfer graphs, 
which are used to detect certain patterns called \emph{leaks}.

\subsection{Leaking graphs}
In order to detect whether a play $\rho = S_0 \arr{a_1,G_1} S_1
\arr{a_2,G_2}\ldots$ has finite capacity, it is enough to detect
\emph{leaking} graphs (characterising entries of accumulators). Further, leaking graphs have special
\emph{separation} properties which will allow us to 
track a small number of graphs. For $G,H$ two graphs, we denote $(a,b) \in G \cdot H$ iff there exists $z$ with $(a,z) \in G,$ and $(z,b) \in H$.

\begin{defi}[Leaks and separations]
  Let $G,H$ be two transfer graphs.  We say that $G$ \emph{leaks at
    $H$} if there exist states $q,x,y$ with $(q,y) \in G \cdot H$,
  $(x,y) \in H$ and $(q,x) \notin G$.  We say that $G$
  \emph{separates} a pair of states $(r,t)$ if there exists
  $q \in \states$ with $(q,r)\in G$ and $(q,t)\not\in G$. Denote by
  $\Sep(G)$ the set of all pairs $(r,t)$ which are separated by $G$.
\end{defi}
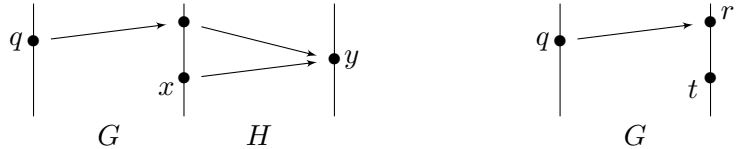
\begin{figure}[htbp]
\begin{center}
\begin{tikzpicture}
\draw[thin] (0,.5) -- (0,-1);
\draw (2,.5) -- (2,-1);
\draw (4,.5) -- (4,-1);
\draw (0,0) node (q) {$\bullet$};
\node  at (q.180) {$q$}; 
\draw(2,-.5) node (x) {$\bullet$};
\node  at (x.210) {$x$}; 
\draw(2,.25) node (x') {$\bullet$};
\draw(4,-.25) node (y) {$\bullet$};
\node  at (y.0) {$y$}; 
\draw[-latex'] (q) -- (x');
\draw[-latex'] (x') -- (y);
\draw[-latex'] (x) -- (y);

\node at (1, -1.25) (G) {$G$};
\node at (3, -1.25) (H) {$H$};


\draw[thin] (7,.5) -- (7,-1);
\draw (9,.5) -- (9,-1);
\draw (7,0) node (q) {$\bullet$};
\node  at (q.180) {$q$}; 
\draw(9,-.5) node (x) {$\bullet$};
\node  at (x.210) {$t$}; 
\draw(9,.25) node (x') {$\bullet$};
\node  at (x'.30)  {$r$}; 
\draw[-latex'] (q) -- (x');
\node at (8, -1.25) (G) {$G$};

\end{tikzpicture}
 \end{center}
\caption{Left: $G$ leaks at $H$; Right: $G$ separates $(r,t)$.}
\end{figure}

The tracking list will be composed of concatenated graphs {\em tracking i} of the form $G[i,j]= G_{i+1} \cdots G_j$ relating $S_i$ with $S_j$:
$(s_i,s_j) \in G[i,j]$ if there exists $(s_k)_{i < k < j}$
with $(s_k,s_{k+1}) \in G_{k+1}$ for all $i \leq k < j$. 
Infinite capacity relates to leaks in the following way:
\begin{lem}
\label{lemma.leaks}
A play has infinite capacity iff there exists an index $i$ 
such that $G[i,j]$ leaks at $G_{j+1}$ for infinitely many indices $j$.
\end{lem}

\begin{proof}
  To prove the right-to-left implication, assume that there exists an
  index $i$ such that $G[i,j]$ leaks at $G_{j{+}1}$ for an infinite
  number of indices $j$.  As the number of states is finite, there
  exists a state $q$ with an infinite number of indices $j$ such that
  we have some $(x_j,y_{j+1}) \in G_{j{+}1}$ with $(q,y_{j+1}) \in G[i,j{+}1]$, $(q,x_j) \notin G[i,j]$.  The accumulator generated by $T_i = \{q\}$ has an infinite number of entries, and we are done with this direction.

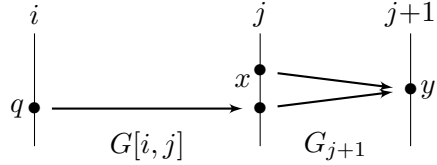
\begin{center}
\begin{tikzpicture}
\node at (0,.75) {$i$};
\node at (3,.75) {$j$};
\node at (5,.75) {$j{+}1$};

\draw (0,.5) -- (0,-1);
\draw (3,.5) -- (3,-1);
\draw (5,.5) -- (5,-1);

\draw (0,-.5) node (q) {$\bullet$};
\node  at (q.180) {$q$}; 

\draw(3,-.5) node (x) {$\bullet$};
\draw(3,0) node (x') {$\bullet$};
\node  at (x'.210) {$x$}; 
\draw(5,-.25) node (y) {$\bullet$};
\node  at (y.0) {$y$}; 
\draw[-latex',thick] (q) -- (x);
\draw[-latex',thick] (x') -- (y);
\draw[-latex',thick] (x) -- (y);

\node at (4,-1) {$G_{j{+}1}$};
\node at (1.5,-1) {$G[i,j]$};
\end{tikzpicture}
\end{center}

\medskip

For the left-to-right implication, assume that there is an accumulator
$(T_j)_{j \geq 0}$ with an infinite number of entries.  

\noindent 
For $X=(X_n)_{n \in\nats}$ an accumulator, we denote $|X_n|$ the number of states in $X_n$, and we define the \emph{width} of $X$ as $\width(X) = \limsup_n |X_n|$. We first prove
the following property:

\medskip

$(\dagger)$ If $\emptyset \neq Y \subseteq X$ and $Z \subseteq X$ are two disjoint accumulators, then
$\width(Z) < \width(X)$.

\medskip

Let us prove property $(\dagger)$.  Let $r$ be the minimal index
s.t. $Y_r \neq \emptyset$. Thus, for every $n\geq r$, $Y_n$ contains at least one vertex.
Because $Y$ and $Z$ are disjoint, we derive
$|Z_n| + 1 \leq |X_n|$. Taking the limsup of this inequality we
obtain $(\dagger)$.

We pick $X$ an accumulator 
of minimal width  
with infinitely many
incoming edges.
Let $r$ minimal such that $X_r \neq \emptyset$.
Let $v \in X_r$.
We denote $Y^v$ the smallest accumulator containing $v$. 
We have $Y^v \neq \emptyset$.
Let us show that $Y^v$ has infinitely many
incoming edges. 
Define 
$Y^v \subseteq T^v = (T^v_n)_{n \in \nats}$ the set of predecessors of vertices
in $Y^v$. We let $Z_n = X_n \setminus T^v_n$ for all $n$. 
We have $Z=(Z_n)$ is an accumulator,
because $T^v$ is predecessor-closed and $X$ is successor-closed.
Applying property $(\dagger)$ to $0 \neq Y^v\subseteq X$ and
$Z \subseteq X$, we obtain $\width(Z) < \width(X)$.  By width
minimality of $X$ among successor-closed sets with infinitely many
incoming edges, $Z = X \setminus T^v$ must have finitely many incoming edges only.
Since $X$ has infinitely many incoming
edges, then $T^v$ has infinitely many incoming edges.  Thus there are
infinitely many edges connecting a vertex outside $Y^v$ to a vertex
of $Y^v$, so that $Y^v$ has infinitely many incoming edges. 
We have just shown that $G[r, j]$ leaks at infinitely many indices $j$.
\end{proof}

Indices $i$ such that $G[i,j]$ leaks at $G_{j+1}$ for infinitely many indices $j$ are said to {\em leak infinitely often}.
Note that if $G$ separates $(r,t)$, and $r, t$ have a common successor
in $H$, then $G$ leaks at $H$.  To link leaks with separations, we
consider for each index $k$, the pairs of states that have a common
successor, in possibly several steps, as expressed by the symmetric
relation $R_k$: $(r,t) \in R_k$ iff there exists $j\geq k+1$ and
$y \in \states$ such that $(r,y) \in G[k,j]$ and $(t,y) \in G[k,j]$.
%
%
%

\begin{lem}
\label{lemma.separation}
For $i<n$ two indices, the following properties hold: 
\begin{enumerate}
\item If $G[i,n]$ separates $(r,t)\in R_{n}$, then there exists
  $m \geq n$ such that $G[i,m]$ leaks at $G_{m+1}$.\label{lem:sep.leak}
\item If index $i$ does not leak infinitely often, then the number of
  indices $j$ such that $G[i,j]$ separates some $(r,t)\in R_{j}$ is 
  finite.\label{lem:sep.finite}
\item If index $i$ leaks infinitely often, then for all $j>i$, 
  $G[i,j]$ separates some $(r,t)\in R_{j}$.\label{lem:sep.always}
\item If $i < j < n$ then $\Sep(G[i, n]) \subseteq \Sep(G[j, n])$.\label{lem:sep.monotone}
\end{enumerate}
\end{lem}

\begin{proof}
  We start with the proof of the first item.  Assume that $G[i,n]$
  separates a pair $(r,t) \in R_n$.  Hence there exists $q$ such that
  $(q,r) \in G[i,n]$, $(q,t) \notin G[i,n]$.  Since $(r,t)\in R_n$, there is an index $k > n$ and a state
  $y$ such that $(r,y) \in G[n,k]$ and $(t,y) \in G[n,k]$.  Hence,
  there exists a path $(t_j)_{n \leq j \leq k}$ with $t_n=t$, $t_{k}=y$,
  and $(t_j,t_{j{+}1}) \in G_{j+1}$ for all $n \leq j < k$.  Moreover,
  there is a path from $q$ to $y$ because there are paths from $q$ to
  $r$ and from $r$ to $y$.  Let $\ell \leq k$ be the minimum index
  such that there is a path from $q$ to $t_\ell$.  As there is no path
  from $q$ to $t_n=t$, necessarily $\ell \geq n+1$. Obviously,
  $(t_{\ell-1},t_{\ell}) \in G_{\ell}$, and by definition and
  minimality of $\ell$, $(q,t_{\ell-1}) \notin G[i,\ell-1]$ and
  $(q,t_{\ell}) \in G[i,\ell]$. That is, $G[i,\ell-1]$ leaks at
  $G_{\ell}$.

\medskip

Let us now prove the second item, using the first one. Assume that
$i$ does not leak infinitely often, and towards a contradiction
suppose that there are infinitely many $j$'s such that $G[i,j]$
separates some $(r,t) \in R_{j}$. To each of these separations, we
can apply item {\bf \emph{1.}} to obtain infinitely many indices
$m$ such that $G[i,m]$ leaks at $G_{m+1}$, a contradiction.

\medskip

We now prove the third item. Since there are finitely many states in
$\states$, there exists $\state \in \states$ and an infinite set $J$
of indices such that for every $j \in J$, $(q,y_{j{+}1}) \in G[i,j{+}1]$,
$(q,x_j) \notin G[i,j]$, and $(x_j,y_{j{+}1}) \in G_{j{+}1}$ for some
$x_j,y_{j{+}1}$. The path from $q$ to $y_{j{+}1}$ implies the existence of
$y_j$ with $(q,y_j) \in G[i,j]$, and $(y_j,y_{j{+}1}) \in G_{j{+}1}$. We have thus found separated pairs $(x_j,y_j) \in R_j$ for every $j \in J$. To exhibit
separations at other indices $k >j$ with $k \notin J$, the natural
idea is to consider predecessors of the $x_j$'s and $y_j$'s.

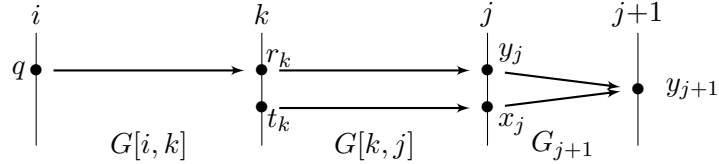
\begin{center}
\begin{tikzpicture}
\node at (0,.75) {$i$};
\node at (3,.75) {$k$};
\node at (6,.75) {$j$};
\node at (8,.75) {$j{+}1$};

\draw (0,.5) -- (0,-1);
\draw (3,.5) -- (3,-1);
\draw (6,.5) -- (6,-1);
\draw (8,.5) -- (8,-1);

\draw (0,0) node (q) {$\bullet$};
\node  at (q.180) {$q$}; 

\draw (3,0) node (q') {$\bullet$};
\node  at (q'.35) {$r_k$}; 

\draw (3,-.5) node (tk) {$\bullet$};
\node  at (tk.325) {$t_k$}; 

\draw(6,-.5) node (x) {$\bullet$};
\node  at ([shift={(.5:.2)}]x.300) {$x_j$}; 
\draw(6,0) node (x') {$\bullet$};
\node  at ([shift={(-.5:.4)}]x'.110)  {$y_j$}; 
\draw(8,-.25) node (y) {$\bullet$};
\node  at ([shift={(.5:.5)}]y.0) {$y_{j{+}1}$}; 
\draw[-latex',thick] (q) -- (q');
\draw[-latex',thick] (q') -- (x');
\draw[-latex',thick] (x') -- (y);
\draw[-latex',thick] (x) -- (y);
\draw[-latex',thick] (tk) -- (x);

\node at (1.5,-1) {$G[i,k]$};
\node at (4.5,-1) {$G[k,j]$};
\node at (7,-1) {$G_{j+1}$};
\end{tikzpicture}
\end{center}

\noindent We define sequences $(r_k,t_k)_{k \geq i}$ inductively as follows.  To
define $r_k$, we take a $j \geq k+1$ such that $j \in J$; this is
always possible as $J$ is infinite.  There exists a state $r_k$
such that $(q,r_k) \in G[i,k]$ and $(r_k,y_j) \in G[k,j]$.

\noindent Also, as $x_j$ belongs to $\im(G[1,j])$, there must
exist a state $t_k$ such that $(t_k,x_j)\in G[k,j]$.  Clearly, $(q,
t_k) \notin G[i,k]$, else $(q,x_j) \in G[i,j]$, which is not
true.  Last, $y_{j{+}1}$ is a common successor of $t_k$ and
$r_k$, that is $(t_k,y_{j{+}1}) \in G[k,j+1]$ and
$(r_k,y_{j{+}1}) \in G[k,j+1]$. Hence $G[i,k]$ separates
$(r_{k},t_{k}) \in R_{k}$.

\medskip

For the last item, let $(r, t) \in \Sep(G[i, n])$ and pick
$q \in \states$ such that $(q, r) \in G[i, n]$ but
$(q, t) \notin G[i, n]$. Since $(q, r) \in G[i, n]$ there exists
$q' \in \states$ so that $(q, q') \in G[i, j]$ and
$(q', r) \in G[j, n]$. It also follows that $(q', t) \notin G[j, n]$
since $(q,q') \in G[i,j]$ but $(q, t) \notin G[i, n]$. Thus we have
shown $(r, t) \in \Sep(G[j, n])$.
\end{proof}

\subsection{The tracking list}
\label{subsec:tl}

Given a fixed index $i$, figuring out whether $i$ is leaking or not
can be done using a deterministic automaton.  However, when one wants
to decide the existence of \emph{some} index that leaks, naively, one
would have to keep track of runs starting from all possible indices
$i \in \nats$. The tracking list will allow us to track only
quadratically many indices at once. The \emph{tracking list} exploits
the relationship between leaks and separations. It is a list of
transfer graphs which altogether separate all possible pairs of
states\footnote{It is sufficient to consider pairs in $R_j$.
However, as $R_j$ is not known \emph{a priori}, 
we consider all  pairs in $Q^2$.}, and are sufficient to detect when leaks occur.

By item (\ref{lem:sep.monotone}) in Lemma~\ref{lemma.leaks}, for any
$n$, $\Sep(G[1, n]) \subseteq \Sep(G[2, n]) \subseteq \cdots
\Sep(G[n,n])$. The {\em exact} tracking list $\mathcal{L}_n$ at step
$n$ is defined as a list of $k \leq |\states|^2$ graphs
$G[i_1,n], \cdots, G[i_k,n]$, where
$1 \leq i_1 < i_2 < \cdots < i_k \leq n$ is the list of indices for
which $\Sep(G[i-1, n]) \neq \Sep(G[i, n])$ (with the convention that
$\Sep(G[0, n]) = \emptyset$). 

\medskip

Consider the sequence of graphs in Figure~\ref{f.ex}, obtained from
alternating {\tt try} and {\tt retry} in the example from
Figure~\ref{fig:tryretry} where \playertwo\ splits agents whenever
possible.
\begin{figure}[htbp]
\centering
\begin{tikzpicture}
    \node at (0,.75) {$0$};
    \node at (2,.75) {$1$};
    \node at (4,.75) {$2$};
    \node at (6,.75) {$3$};
    \node at (8,.75) {$4$};
    \node at (10,.75) {$5$};
    
    \node at (1,.75) {\tt try};
    \node at (5,.75) {\tt try};
    \node at (9,.75) {\tt try};
    \node at (3,.75) {\tt retry};
    \node at (7,.75) {\tt retry};
    
    \node at (1,-2.15) {$G_1$};
    \node at (3,-2.15) {$G_2$};
    \node at (5,-2.15) {$G_3$};
    \node at (7,-2.15) {$G_4$};
    \node at (9,-2.15) {$G_5$};
    
    \draw (0,.5) -- (0,-2);
    \draw (2,.5) -- (2,-2);
    \draw (4,.5) -- (4,-2);
    \draw (6,.5) -- (6,-2);
    \draw (8,.5) -- (8,-2);
    \draw (10,.5) -- (10,-2);
    
    \draw (0,-0.5) node (q1) {$\bullet$};
    \node  at (q1.225) {$q_0$}; 
    
    \draw (4,-0.5) node (q2) {$\bullet$};
    \node  at (q2.225) {$q_0$}; 
    
    \draw (8,-0.5) node (q3) {$\bullet$};
    \node  at (q3.225) {$q_0$}; 
    
    \draw (2,0) node (t1) {$\bullet$};
    \node  at (t1.35) {$q_\top$}; 
    
    \draw (6,0) node (t2) {$\bullet$};
    \node  at (t2.35) {$q_\top$}; 
    
    \draw (10,0) node (t3) {$\bullet$};
    \node  at (t3.35) {$q_\top$}; 

    \draw (2,-1) node (b1) {$\bullet$};
    \node  at (b1.35) {$q_\bot$}; 
    
    \draw (6,-1) node (b2) {$\bullet$};
    \node  at (b2.35) {$q_\bot$}; 
    
    \draw (10,-1) node (b3) {$\bullet$};
    \node  at (b3.35) {$q_\bot$};
    
    \draw (4,-1.5) node (s1) {$\bullet$};
    \node  at (s1.300) {$k$};
    
    \draw (6,-1.5) node (s2) {$\bullet$};
    \node  at (s2.300) {$k$};
    
    \draw (8,-1.5) node (s3) {$\bullet$};
    \node  at (s3.300) {$k$};
    
    \draw (10,-1.5) node (s4) {$\bullet$};
    \node  at (s4.300) {$k$};

    \draw[-latex',thick] (q1) -- (t1);
    \draw[-latex',thick] (q1) -- (b1);
    \draw[-latex',thick] (t1) -- (q2);
    \draw[-latex',thick] (b1) -- (s1);
    \draw[-latex',thick] (q2) -- (t2);
    \draw[-latex',thick] (q2) -- (b2);
    \draw[-latex',thick] (s1) -- (s2);
    \draw[-latex',thick] (t2) -- (q3);
    \draw[-latex',thick] (b2) -- (s3);
    \draw[-latex',thick] (s2) -- (s3);
    \draw[-latex',thick] (q3) -- (t3);
    \draw[-latex',thick] (q3) -- (b3);
    \draw[-latex',thick] (s3) -- (s4);

\end{tikzpicture}
\caption{Sequence of graphs associated with a run.}
\label{f.ex}
\end{figure}
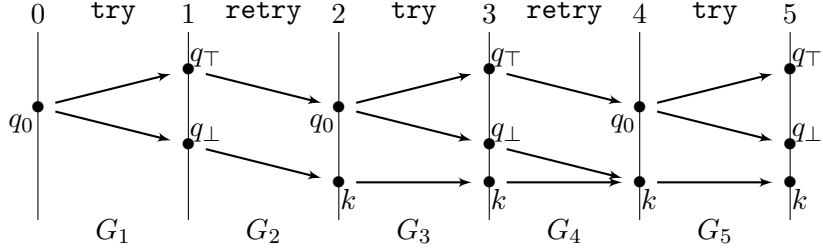

Let us compute $\Sep(G[i,5])$ on that example for $0 \leq i\leq
4$. 
The graph $G[0,5]=G_1 \cdots G_5$ has the following edges:
$(\stateinit,q_\top),(\stateinit,q_\bot), (\stateinit,k)$, and thus
$\Sep(G[0,5]) = \emptyset$.  In comparison,
$\Sep(G[1,5]) = \{ (k,q_\top), (k, q_\bot) \}$ because $(q_\bot,k)$ is
an edge of $G[1,5]$ and $(q_\bot,q_\top)$ and $(q_\bot,q_\bot)$ are
not.  Also, $\Sep(G[2,5]) = \{ (k,q_\top), (k,q_\bot) \}$.  Finally
$\Sep(G[3,5]) = \Sep(G[4,5]) = \{ (k,q_\top), (k,q_\bot), (q_\top,k),
(q_\bot,k)\}$. Thus, $\mathcal{L}_5 = (G[1,5];G[3,5])$.  Notice that
$G[1,5]= \{(q_\top,q_\top),(q_\top,q_\bot), (q_\top,k),(q_\bot,k)\}$
and $G[3,5]= \{(q_\top,q_\top),(q_\top,q_\bot), (q_\bot,k),(k,k)\}$.

\medskip

The {\em exact} tracking list ${\mathcal{L}}_n$ allows one to test for
infinite leaks, but computing it with polynomial memory seems
hard. Instead, we propose to approximate the {\em exact} tracking list into a list,
namely the tracking list $\overline{\mathcal{L}}_n$,
which needs only polynomial memory to be computed, and which is
sufficient for finding infinite leaks.

\medskip

The tracking list $\overline{\mathcal{L}}_n$ is also of the form
$\{ G[i_1, n], G[i_2, n] , \ldots G[i_k, n] \}$ where
$0 \leq i_1 < i_2 \ldots i_k < n$ with
$\emptyset \neq \Sep(G[i_r, n]) \subsetneq \Sep(G[i_{r+1}, n])$. It is
is computed inductively in the following way:
$\overline{\mathcal{L}}_0$ is the empty list.
 For $n > 0$, 
the list $\overline{\mathcal{L}}_{n}$ is computed 
from $\overline{\mathcal{L}}_{n-1}$ and $G_n$
in three stages by the following
update$\_$list algorithm:
\begin{enumerate}
\item First, every graph $G[i,n-1]$ in the list $\overline{\mathcal{L}}_{n-1}$ is
  concatenated with $G_{n}$, yielding $G[i,n]$.

\item Second, $G_{n}=G[n-1,n]$ is added at the end of the list.
 
\item Lastly, the list is filtered: a graph $H$ is kept if and only if
  it separates a pair of states $(p,q)\in \states^2$ which is not separated by
  any graph that appears earlier in the list.\footnote{This algorithm can be performed without knowing the indices $(i_j)_{j \leq k}$, but just the graphs $(G[i_j, n])_{j \leq k}$.}
\end{enumerate}
Under this definition,
$\overline{\mathcal{L}}_{n} = \{ G[i_j, n] \mid 1 \leq j \leq k, \,
\Sep(G[i_{j-1}, n]) \neq \Sep(G[i_j, n]) \}$, with the convention
that $\Sep(G[i_0, n]) = \emptyset$.

\medskip

Notice that the tracking list $\overline{\mathcal{L}}_n$ may differ from  the exact tracking list $\mathcal{L}_n$, as shown with 
the example on Figure \ref{f.tracking}. 
We have $\mathcal{L}_3 = (G[1,3],G[2,3])$,
as $G[0,3]=G_1 \cdots G_3$ does not
separate any pair of states,
$G[1,3]= G_2 \cdot G_3$ separates $(q_1,q_2)$ and 
$G[2,3]= G_3$ separates $(q_1,q_2)$ and $(q_2,q_1)$.
However, 
$\overline{\mathcal{L}}_3 = (G[2,3]) \neq \mathcal{L}_3$.
Indeed, 
$\mathcal{L}_2 = \overline{\mathcal{L}}_2 = (G[0,2])$
as $G[0,2]=G_1 \cdot G_2$ and $G[1,2]=G_2$
separate exactly the same pairs $(q_1,q_2);(q_1,q_3);(q_2,q_3);(q_3,q_2)$.
Applying the update$\_$list algorithm, we obtain the 
intermediate list 
(G[0,3],G[2,3])
after stage 2. As $G[0,3]$ separates no pair of states, it is filtered out in stage 3. We obtain 
$\overline{\mathcal{L}}_3 = (G[2,3]) \neq \mathcal{L}_3$.

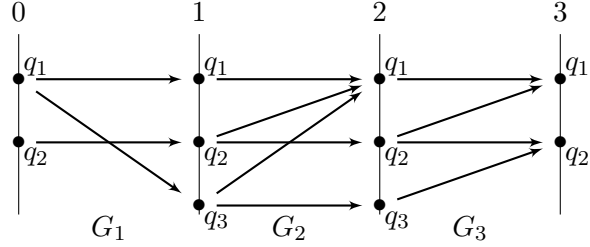
\begin{figure}[t!]
\begin{center}
\begin{tikzpicture}[scale=1.2]
    \node at (0,.75) {$0$};
    \node at (2,.75) {$1$};
    \node at (4,.75) {$2$};
    \node at (6,.75) {$3$};
    
    \draw (0,.5) -- (0,-1.5);
    \draw (2,.5) -- (2,-1.5);
    \draw (4,.5) -- (4,-1.5);
    \draw (6,.5) -- (6,-1.5);
    
    \draw (0,0) node (t1) {$\bullet$};
    \node  at (t1.35) {$q_1$}; 
    
    \draw (2,0) node (t2) {$\bullet$};
    \node  at (t2.35) {$q_1$}; 
    
    \draw (4,0) node (t3) {$\bullet$};
    \node  at (t3.35) {$q_1$}; 

    \draw (6,0) node (t4) {$\bullet$};
    \node  at (t4.35) {$q_1$}; 

    \draw (0,-0.7) node (u1) {$\bullet$};
    \node  at (u1.320) {$q_2$}; 
    
    \draw (2,-0.7) node (u2) {$\bullet$};
    \node  at (u2.320) {$q_2$}; 
    
    \draw (4,-0.7) node (u3) {$\bullet$};
    \node  at (u3.320) {$q_2$}; 

    \draw (6,-0.7) node (u4) {$\bullet$};
    \node  at (u4.320) {$q_2$};

    \draw (2,-1.4) node (v1) {$\bullet$};
    \node  at (v1.320) {$q_3$}; 

    \draw (4,-1.4) node (v2) {$\bullet$};
    \node  at (v2.320) {$q_3$};

    \draw[-latex',thick] (t1) -- (t2);
    \draw[-latex',thick] (t2) -- (t3);
    \draw[-latex',thick] (t3) -- (t4);
    
    \draw[-latex',thick] (u1) -- (u2);
    \draw[-latex',thick] (u2) -- (u3);
    \draw[-latex',thick] (u3) -- (u4);
    
    \draw[-latex',thick] (v1) -- (v2);
    
    \draw[-latex',thick] (t1) -- (v1);
    \draw[-latex',thick] (v1) -- (t3);
    \draw[-latex',thick] (v2) -- (u4);
    \draw[-latex',thick] (u2) -- (t3);
    \draw[-latex',thick] (u3) -- (t4);

    \node at (1,-1.65) {$G_1$};
    \node at (3,-1.65) {$G_2$};
    \node at (5,-1.65) {$G_3$};
 \end{tikzpicture}
\caption{Example where the tracking list $\overline{\mathcal{L}}_3$
  differs from the exact tracking list ${\mathcal{L}}_3$.}
\label{f.tracking}
\end{center}
\end{figure}

Let $\overline{\mathcal{L}}_n = \{H_1 ,\cdots ,H_\ell\}$ be the tracking list at step $n$. Each transfer graph $H_r \in \overline{\mathcal{L}}_n$ is of
the form $H_r = G[t_r,n]$. We say that $r$ is the \emph{level} of
$H_r$, and $t_r$ the \emph{index tracked} by $H_r$. Observe that the
lower the level of a graph in the list, the smaller the index it
tracks.


%
%

When we consider the sequence of tracking lists
$(\overline{\mathcal{L}}_n)_{n \in \mathbb{N}}$, for every index $i$,
either it eventually stops to be tracked or it is tracked forever from
step $i$, \emph{i.e.}\ for every $n \geq i$, $G[i,n]$ is not filtered
out from $\overline{\mathcal{L}}_n$. In the latter case, $i$ is said
to be \emph{remanent} (it will never disappear).

\begin{lem}
\label{lem:caracPG}
A play has infinite capacity iff there exists an index $i$ such that
$i$ is remanent and leaks infinitely often.
\end{lem}
\begin{proof}
Because of Lemma~\ref{lemma.leaks} we only need to show that if there is an
index $i$ that leaks infinitely often, then there is an index which is
remanent and leaks infinitely often.

Let $i$ be the smallest index that leaks infinitely often. By
Lemma~\ref{lemma.separation} (\ref{lem:sep.finite}) there is an $N > i$ so
that whenever $j < i < N \leq n$, $\Sep(G[j, n]) \cap R_n = \emptyset$.
Similarly, by Lemma~\ref{lemma.separation} (\ref{lem:sep.always}), for every $n > i$, $\Sep(G[i, n]) \cap R_n \neq
\emptyset$. Combined with Lemma~\ref{lemma.separation} (\ref{lem:sep.monotone}), this implies that there is some index 
$j^* \geq i$ which is remanent.
Indeed, let $\overline{\mathcal{L}}_N = \{ G[i_1, N],
G[i_2, N] \ldots G[i_k, N]\}$ and let $j^* = \min \{ i_r \mid i_r \geq i\}$.
Index $j^*$ exists because for any $i_r < i$, $\Sep(G[i_r, N])  \subsetneq \Sep(G[i, N]) \subseteq \Sep(G[N-1, N])$. The strict inequality arises because $\Sep(G[i_r, N]) \cap R_N = \emptyset$ but $\Sep(G[i, N]) \cap R_N \neq \emptyset$. By the same argument, for every $n \geq N$ and $i_r < i$, $\Sep(G[i_r, n]) \cap R_n = \emptyset$
and $\Sep(G[j^*, n]) \cap R_n \supseteq \Sep(G[i, n]) \cap R_n \neq \emptyset$.
This shows that $j^*$ is remanent. By Lemma~\ref{lemma.separation}
(\ref{lem:sep.finite}), index $j^*$ also leaks infinitely often.
\end{proof}

\subsection{The parity game}
We now describe a parity game $\PG$, which extends the support arena with 
on-the-fly computation of the tracking list.

{\medskip \noindent \bf Priorities.} By convention, lowest priorities are the most important
and the odd parity is good for \playerone, so \playerone\ wins iff the $\liminf$ of the priorities
is odd. With each level $1\leq r \leq |\stnb|^2$ of the tracking list are associated two priorities
$2r$ (graph $G[i_r,n]$ non-remanent) and $2r+1$ (graph $G[i_r,n]$ leaking), and on top of that are added priorities $1$ (goal reached) and $2 |\stnb|^2+2$ (nothing),
hence the set of all priorities is $\{1,\ldots,2 |\stnb|^2+2\}$. 

When \playertwo\ chooses a transition labelled by a transfer graph
$G$, the tracking list is updated with $G$ and the priority of the
transition is determined as the smallest among: priority 1 if the
support  $\{f\}$ has ever been visited, priority $2r+1$ for the
smallest $r$ such that $H_r$ (from level $r$) leaks at $G$,
priority $2r$ for the smallest level $r$ where graph $H_r$ was removed
from ${\mathcal L}$,
and in all other cases priority $2 |\stnb|^2+2$.
%
%

{\medskip \noindent \bf  States and transitions.}
$\GG^{\leq |\stnb|^2}$ denotes the set of list of at most $|\stnb|^2$ transfer graphs.
\begin{itemize}
\item States of $\PG$ form a subset of
  $\{0,1\} \times 2^Q\times \GG^{\leq |\stnb|^2}$, each state being of
  the form $(b,S,H_1,\ldots,H_\ell)$ with $b \in \{0,1\}$ a bit
  indicating whether the support $\{f\}$ has been seen, $S$ the
  current support and $(H_1,\ldots,H_\ell)$ the tracking list.
%
The initial state is $(0,\{\state_0\},\varepsilon)$, with $\varepsilon$ the empty list.

\item Transitions in $\PG$ are all $(b,S,H_1,\ldots,H_{\ell})
  \arr{\priority,a,G} (b',S',H'_1,\ldots,H'_{\ell'})$ where $\priority$
  is the priority, and such that $S \arr{a,G} S'$ is a transition of
  the support arena, and
  \begin{enumerate}
  \item $(H'_1,\ldots,H'_{\ell'})$ is the tracking list obtained by
    updating the tracking list $(H_1,\ldots,H_{\ell})$ with $G$, as
    explained in subsection~\ref{subsec:tl};
  \item if $b=1$ or if $S' = \{\targetstate\}$, then $\priority =1$ and
    $b'=1$;
\item otherwise $b'=0$.
In order to compute the priority $\priority$, we let $\priority'$ be the smallest level $1\leq r \leq \ell$ such that $H_r$ 
leaks at $G$ and $\priority'=\ell+1$ if there is no such level, and we
also let
$\priority''$ as the minimal level $1\leq r \leq \ell$ such that $H'_{r}
\neq H_{r} \cdot G$ and $\priority''=\ell+1$ if there is no such
level. Then $\priority=\min( 2 \priority'+1, 2 
\priority'')$.  
\end{enumerate}
\end{itemize}


We are ready to state the main result of this paper, which yields an
\EXPTIME complexity for the population control problem. This entails
the first statement of Theorem~\ref{th1}, and together with
Proposition~\ref{th.2wins}, also the first statement of Theorem~\ref{th2}.

\begin{thm}
\label{th.parity}
  \playerone\ wins the game $\PG$ if and only if \playerone\ wins the
  \ror\ game.
  Solving these games can be done in time
  $O(2^{(1+|\stnb|+|\stnb|^4)(2|\stnb|^2+2)})$. 
Strategies with $2^{|\stnb|^4}$ 
memory states are sufficient to both \playerone\ and \playertwo.
\end{thm}

\begin{proof}
{\color{black} 
The state space of parity game $\PG$ is the product of the 
set of supports  
with a deterministic automaton computing the tracking list.
There is a natural correspondence between plays and strategies in the parity game $\PG$
and in the \ror\ game.

\playerone\ can win the parity game $\PG$ in two ways: either the play
visits the support $\{f\}$, or the priority of the play is $2r+1$ for
some level $1\leq r \leq |Q|^2$. By design of $\PG$, this second
possibility occurs iff $r$ is remanent and leaks infinitely often.
According to Lemma~\ref{lem:caracPG}, this occurs if and only if the
corresponding play of the \ror\ game has infinite capacity.  Thus
\playerone\ wins $\PG$ iff he wins the capacity game.

In the parity game $\PG$, there are at most $2^{1 +
  |\stnb|}\left(2^{|\stnb|^2}\right)^{|Q|^2}=2^{1 + |\stnb| +
  |\stnb|^4}$ states and $2|Q|^2+2$ priorities, implying the
complexity bound using state-of-the-art
algorithms~\cite{JurdzinskiStacs2000}.  
Actually the complexity is even quasi-polynomial according to the
algorithms in~\cite{sanjay}.  Notice however that this has little impact on the complexity of the population control problem, as the number of priorities is logarithmic in the number of states of our parity game.

Further, it is well known that the winner of a parity game has a
positional winning strategy~\cite{JurdzinskiStacs2000}. 
A \emph{positional} winning strategy $\sigma$ in the game $\PG$
corresponds to a \emph{finite-memory} winning strategy $\sigma'$ in
the \ror\ game, whose memory states are the states of $\PG$. Actually
in order to play $\sigma'$, it is enough to remember the tracking list,
\emph{i.e.}\ the third component of the state space of $\PG$. Indeed,
the second component, in $2^Q$, is redundant with the actual state of
the \ror\ game and the bit in the first component is set to $1$ when
the play visits $\{f\}$ but in this case the \ror\ game is won by
\playerone\ whatever is played afterwards. Since there are at most
$2^{|Q|^4}$ different tracking lists, we get the upper bound on the
memory.} { }
\end{proof}

\section{Number of steps before Synchronization}
\label{sec:timetosynch}
In this section, we will restrict ourselves to {\em controllable NFAs}, that is positive instances of the population control problem.
To be useful in the presence of many agents, the controller should be able to gather all agents in the target state in a reasonable time (counted as the number of actions played before synchronization).
Notice that this notion is similar to the termination time used for population protocols in \cite{BEK18}.


\subsection{Dependency with respect to the number of agents}
\label{s.number}

We first show that there are controllable NFAs for which \playerone\ requires a
quadratic number of steps (in the number of agents) against the best
strategy of \playertwo. 
Consider again the NFA 
$\cA_{\text{time}}$  
from Figure~\ref{fig:tryretry} (see also Figure~\ref{f.gadget}, left).
Recall that $\cA_{\text{time}}$ is controllable.

%
%
%
%
%
%

\begin{lem}
\label{slow}
For the NFA $\cA_{\text{time}}$, \playerone\ requires $\Theta(m^2)$
steps to win in the worst case.
\end{lem}

\begin{proof}
  A winning strategy $\sigma$ for
  \playerone\ is to play {\tt try} followed by {\tt keep} until only
  one of $\state_\top$ or $\state_\bot$ is filled, in which case
  action {\tt top} or {\tt bottom} can be played to move the
  associated agents to the target state.
  This will eventually happen,
  as the number of agents in $\state_{\top}$ and in $\state_{\bot}$ is
  decreasing while the number of agents in state $k$ increases upon
  ({\tt try};{\tt keep}) and when $\state_{\bot}$ is not empty. This
  is the strategy generated from our algorithm.
%

  We now argue on the number of steps $\sigma$ needed to send all
  agents to the target state. Observe that the only 
  non-deterministic action is {\tt try} from state $\stateinit$.
  Clearly enough, regarding the number of steps before
  synchronisation, \playertwo' best answer is to move one agent in
  $\state_{\bot}$ and the remaining agents in $\state_{\top}$ upon
  each {\tt try} action. Letting $m$ be the number of agents, the run
  associated with $\sigma$ and the best counterstrategy for
  \playertwo\ is
  \[({\tt try} ; {\tt keep})^{m-1}; {\tt bot} ; {\tt restart }; ({\tt
      try} ; {\tt keep})^{m-2}; {\tt bot}; {\tt restart } \cdots {\tt
      try} ; {\tt keep} ; {\tt bot}\] and its size is
  $\sum_{i=1}^{m-1} (2i+2) -1 = O(m^2)$.  The system thus requires a
  quadratic number of steps before synchronisation, and this is the
  worst case.
\end{proof}

Notice that the above result only needs a fixed number of states,
namely 6.

\tikzset{elliptic state/.style={draw,rectangle}}
\begin{figure}[htbp]
  \begin{tikzpicture}
[-latex',>=stealth',shorten >=1pt,auto,node
    distance=2.2cm]
    \tikzstyle{every state}=[draw=black,text=black, inner
    sep=2pt,minimum size=12pt]



  \node[initial, initial text={},state] (A)                    {$\stateinit$};
  \node[state]         (B) [above right of=A] {$\state_\top$};
  \node[state]         (C) [below right of=A] {$\state_\bot$};
  \node[state]         (D) [left of=C]       {$k$};
  \node[state]         (E) [fill=green!80!black,below right of=B]       {};

  \path (A) edge  node [below right] {$\overline{\tt try}$} (B)
            edge              node [above right]  {$\overline{\tt try}$} (C)
        (B) edge [bend right] node [above left] {$\overline{\tt keep}$} (A)
            edge              node  {$\overline{\tt top}$} (E)
        (C) edge              node {$\overline{\tt keep}$} (D)
            edge              node [below right] {$\overline{\tt bot}$} (E)
        (D) edge              node {$\overline{\tt restart}$} (A)
            edge [loop below] node {$\overline{\actions} \setminus \{\overline{\tt restart}\}$} (D)
        (E) edge [loop right] node {$\overline{\actions}$} (E);


  \node[initial, initial text={},elliptic state, anchor=west, right] at (8,0) (X) [minimum height = 1.5cm, minimum width = 2cm]  {$m^2$-\textsf{Gadget}};
  \node[state]         (Y) [fill=green!80!black, right of=X]       {};

  \path (X) edge [line width = 0.5mm] node [pos=0.2, below] {} (Y)
  ;
        
\end{tikzpicture}
\caption{NFA $\overline{\cA_\text{time}}$ on alphabet $\overline{\actions}$ and its abstraction into the gadget $m^2$-\textsf{Gadget}. 
  }
\label{f.gadget}
\end{figure}
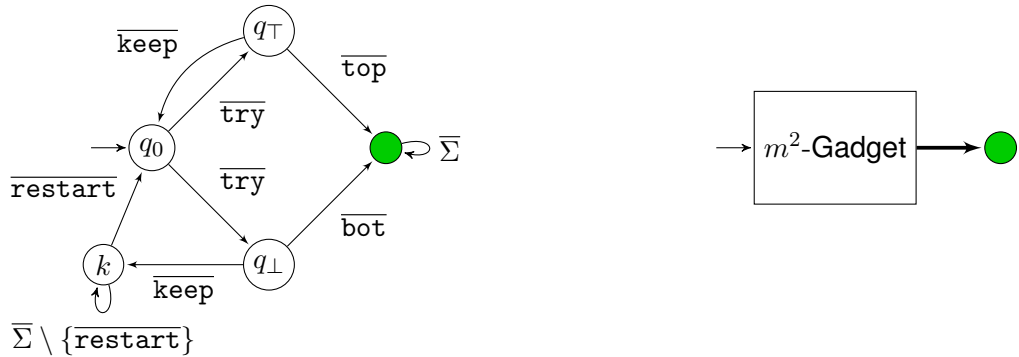

\subsection{Polynomial bound on the number of steps for synchronization}
We now show that we can build an NFA with $n$ states such that the
system requires order $m^{0(n)}$ steps before synchronisation.  For
that, we turn the system $\cA_{\text{time}}$ into a gadget, as shown
on Figure~\ref{f.gadget}.  This gadget will be used in an inductive
manner to obtain 
an NFA for which  
$m^{0(n)}$ steps before 
synchronisation are required against the best strategy of \playertwo.

First, we construct an NFA which requires $O(m^{4})$ steps before
synchronisation. Essentially we replace the edge from $\stateinit$ to
$q_{\top}$ in the NFA $\cA_{\text{time}}$ by the $m^2$-\textsf{Gadget}
to obtain the NFA $\cA_{4}$ on Figure \ref{f.quantic}.
The alphabet $\overline{\actions}$ of actions of the $m^2$-\textsf{Gadget} is a disjoint copy 
of the alphabet $\actions$ of actions of $\cA_{\text{time}}$.
In particular, playing any action of $\Sigma$ when any token is in 
the $m^2$-\textsf{Gadget} leads to the losing sink state $\frownie$.

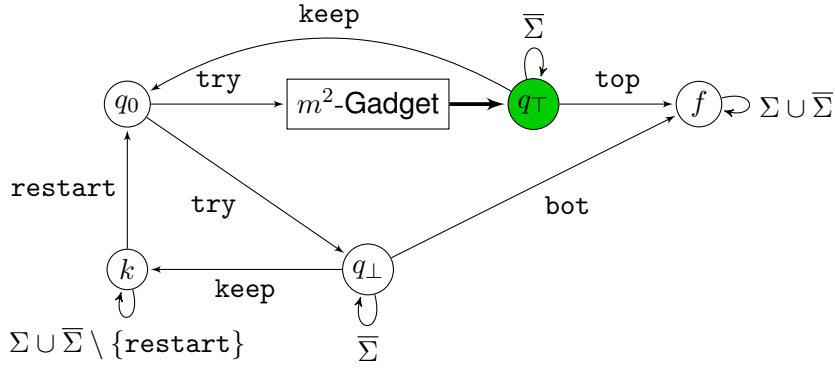
\begin{figure}[htbp]
\centering
\begin{tikzpicture}
[-latex',>=stealth',shorten >=1pt,auto,node
    distance=2.2cm]
    \tikzstyle{every state}=[draw=black,text=black, inner
    sep=2pt,minimum size=12pt]

  \node[state] (A)                    {$\stateinit$};
  \node[elliptic state] (G) [right of=A,xshift=1cm]     {$m^2$-\textsf{Gadget}};
  \node[state]         (B) [fill=green!80!black,right of=G] {$q_{\top}$};
  \node[state]         (C) [below of=G] {$q_{\bot}$};
  \node[state]         (D) [below of=A]       {$k$};
  \node[state]         (E) [right of=B]       {$\targetstate$};

  \path (A) edge 				   node  {\tt try} (G)
 		  (A) edge 				   node [below left] {\tt try} (C)
        (G) edge      [line width = 0.5mm]        node  {} (B)
        (B) edge [bend right] node [above] {\tt keep} (A)
        (B)    edge   node [pos=.5,above] {\tt top} (E)
		         edge [loop above] node {$\overline{\actions}$} (B)

        (C) edge              node  {\tt keep} (D)
            edge 				   node [below right] {\tt bot} (E)
	         edge [loop below] node {$\overline{\actions}$} (C)

        (D) edge              node {\tt restart} (A)
            edge [loop below] node {$\actions \cup \overline{\actions} \setminus \{{\tt restart}\}$} (D)
        (E) edge [loop right] node {$\actions \cup \overline{\actions}$} (E);

\end{tikzpicture}
\caption{The NFA $\cA_4$ which requires $\Theta(m^4)$ steps to
  synchronize all agents into the final state. State $q_{\top}$ is the
  output state of the $m^2$-\textsf{Gadget}.}
\label{f.quantic}
\end{figure}

Consider the strategy of \playertwo\ which is to place 1 agent in 
$q_\bot$ (resp. $\overline{q_\bot}$)
and the rest to $q_\top$ (resp. $\overline{q_\bot}$)
when action {\tt try} (resp. $\overline{\tt try}$)
is played.
Relying on Lemma~\ref{slow}, 
any strategy of \playerone\ needs $m^2$ steps to place 1 agent in $k$, then $(m-1)^2$ steps for the next agent, etc.
Thus, any strategy needs $O(m^3)$ steps to place 1 agent in the target state $\targetstate$. 
Finally, any strategy needs $O(m^4)$ steps to place all the agents in the target state $\targetstate$.  We thus obtain:

\begin{lem}
For the NFA $\cA_4$, \playerone\ requires $\Theta(m^4)$
steps to win in the worst case.
\end{lem}




One can repeat this construction, nesting copies of the
$m^2$-\textsf{Gadget}. At each new gadget, the number of states in the
NFA increases by a constant amount, namely $4$. The $\ell$-layered
NFA, consisting of $\ell-1$ nested gadgets, has $4\ell+2$ states and
requires $\Theta(m^{2\ell})$ steps before synchronisation. We thus
derive the following upper-bound on the time to synchronisation:
\begin{cor}
\label{c.exp}
There exist NFAs with $|Q|$ states such that $m^\frac{|Q|-2}{2}$
steps are required by any strategy to synchronise $m$ agents.
\end{cor}

\subsection{Optimality of the winning strategy}
In this subsection, we show that the winning strategy built by our
algorithm is optimal, in the sense that it never requires more than
$m^{|Q|^{O(1)}}$ steps to synchronize $m$ agents. 
For $\nfa$ a controllable NFA, we write $\sigma_\nfa$ the winning strategy built by our algorithm.
\begin{thm}
\label{th.steps}
For $\nfa$ with $|Q|$ states, $\sigma_\nfa$ needs at most
$m^{|Q|} \times 2^{|Q|^4}$ steps to synchronise $m$ agents in the
target state $\targetstate$.
\end{thm}

To prove Theorem~\ref{th.steps}, an essential ingredient is the following:

\begin{lem}
\label{l.leak}

Let $\Memstate$ be a memory state of $\sigma_\nfa$, $S \neq \{f\}$ a support and
$H$ a transfer graph such that 
(i) $(\Memstate,S)$ is reachable under $\sigma_\nfa$,
(ii) $H$ is compatible with $\sigma_\nfa$ from $(\Memstate,S)$, 
and (iii) $H$ is a loop around $(\Memstate,S)$.
Then, there exists a
partition $S = T \uplus U$ such that $H(U) \subseteq U$ and
$H(T) \cap U \neq \emptyset$, where
$H(X) = \{q \in S \mid \exists p\in X,\ (p,q) \in H\}$.
\end{lem}

Intuitively, letting $G$ a run according to $\sigma_\nfa$
reaching $(\Memstate,S)$, the run $GH^\omega$ is also according to $\sigma_\nfa$.
As $\sigma_\nfa$ is winning, 
$GH^\omega$  must be non realisable, that is 
$GH^\omega$ needs to have an accumulator with an infinite number of entries. We prove in this lemma that for a repeated graph $H$, 
there is a structural characterisation of accumulators with an infinite number of entries: $S$ can be partitioned into $U \uplus T$, where $U$ corresponds to a structural accumulator, and there is one entry from $T$ to $U$ for each $H$.

\begin{proof}
  As \playerone\ is winning, the play $GH^\omega$ must have infinite
  capacity.  By Lemma~\ref{lemma.leaks}, there exist an index $i$ and
  an infinite sequence of indices $j>i$ such that $H^{j-i}$ leaks at
  $H$.  Because the same graph $H$ is repeated, we can assume wlog
  that $i=0$.  As the number of states is finite, there exists a
  triple $(q, x, y )$ such that $(q,y) \in H^{j+1}$, $(x,y) \in H$ and
  $(q,x) \notin H^{j}$ for all $j$'s in an infinite set $J$ of indices.

  \medskip

  Let $\ell$ be the number of steps in $G$. We consider $H$ atomic,
  that is, as if it is a single step.  For every state $s \in S$,
  define $A^s=(A^s_n)_{n \in \nats}$ the smallest accumulator with
  $A^s_\ell=\{s\}$.

If any state $s \in S$ is such that
$s \notin A^{s}_n$
for any $n > \ell$, then we are done as we can set 
$U = \bigcup_{n > \ell} A^{s}_n$, and $T=S \setminus U$.
Indeed, $s \notin U$, and $s$ has a successor in $U$ (any state in $A^{s}_{\ell+1}$).
We can thus assume that 
$s \in \bigcup_{n > \ell} A^{s}_n$
for all $s \in S$.

Let $U = \bigcup_{n > \ell} A^{y}_n = \bigcup_{n \geq \ell} A^{y}_n$
as $y \in \bigcup_{n > \ell} A^{y}_n$ and $A^y_\ell=\{y\}$.
Let $T = S \setminus U$.
We will show that this choice of $(T,U)$
satisfies the condition of the statement.
In particular, we will show that $x \notin U$, 
and as $(x,y) \in H$ and $y \in U$, we are done.

Let $k_x,k_y>0$ such that $(x,x) \in H^{k_x}$,
and $(y,y) \in H^{k_y}$. Let $k = k_x \times k_y$.
Partition $J$ into sets $J_r = \{j \mid j = r (\mod k)\}$, 
for $r < k$.
As $J$ is infinite, one of $J_r$ must be infinite.
Taking two indices $j,j' \in J_r$, we have that
$j'-j$ is a multiple of $k$.
In particular, we have $(x,x)\in H^{j'-j}$
and $(y,y)\in H^{j'-j}$.
We also know that $(y,x) \notin H^{j'-j-1}$
as $(q,x) \notin H^{j'}$
and $(q,y) \in H^{j+1}$.

For any state $s\in S$, let $width(s)= max_n(|A_n^s|)$.
Easily, for all $s \in U$, 
$width(s) \leq width(y)$, 
as $s$ can be reached by $y$, 
let say in $v$ steps, 
and hence $A_n^s \subseteq A_{n+v}^y$.
We now show that $width(x) > width(y)$, which implies that 
$x \notin U$, and we are done.
For all $n$, we have $A_n^y \subseteq A_{n+1}^x$, 
as $(x,y) \in H$. 
Now, there are two indices $j,j' \in J_r$
such that $j'-j > n+1$, as $J_r$ is infinite.
Let $z$ such that $(x,z) \in H^{n+1}$
and $(z,x) \in H^{j'-j-n-1}$, which must exists as $(x,x) \in H^{j'-j}$.
As $(y,x) \notin H^{j'-j-1}$, we have
$(y,z) \notin H^{n}$.
That is, 
$z \in A_{n+1}^x \setminus A_n^y$.
Thus, $|A_{n+1}^x| \geq |A_n^y|+1$, and thus
$width(x) > width(y)$.
\end{proof}

We can now prove Theorem \ref{th.steps}:

\begin{proof} 
  Assume by contradiction that there is a run consistent with
  $\sigma_\nfa$ lasting more than $m^{|Q|} \times 2^{|Q|^4}$ steps
  before synchronisation.  Because there are no more than $2^{|Q|^4}$
  different memory states, there is one memory state $\Memstate$ which
  is repeated at least $m^{|Q|}$ times.  Let us decompose the path as
  $G_1 G_2 \cdots G_{m^{|Q|}}$, such that $G_{i} \ldots G_j$ is a loop
  around $\Memstate$, for all $1 \leq i < j \leq m^{|Q|}$. We write $S$
  for the support associated with $\Memstate$.

\medskip

By Lemma \ref{l.leak} applied with $H=G_1 G_2 \cdots G_{m^{|Q|}}$,
there exists a partition $S = T \uplus U$ , such that
$H (U) \subseteq U$ and $H ( T) \cap U \neq \emptyset$.  We define a
sequence $U_1, \ldots, U_{m^{|Q|}}$ of supports inductively as
follows:
\begin{itemize}
  \item $U_{m^{|Q|}}=U$, and 
  \item for $i=m^{|Q|}$ to $i= 1$,
 $U_{i-1} = \{s \in S \mid G_i(\{s\}) \subseteq U_{i}\}$.
\end{itemize}

We further write $T_i= S \setminus U_i$.
We have 
$U \subseteq U_0 \subsetneq S$,
and $U_i \neq \emptyset$
as $G_{i}(U_{i-1}) \subseteq U_{i}$
for all $i$.
In the same way, $U_i \neq S$
for all $i$ as otherwise we would have $U_{m^{|Q|}}=S$, 
a contradiction with $U_{m^{|Q|}}=U \neq S$.
Hence $T_i \neq \emptyset$  for all $i$.

Consider now the set $K$ of indices $i$ such that
$G_i ( T_{i-1}) \cap U_i \neq \emptyset$. By Lemma~\ref{l.leak}, $K$
is nonempty. Moreover, since there are $m$ agents, $|K| < m$, otherwise,
$T_0$ would contain at least $m$ more agents than $T_{m^{|Q|}}$. Thus
$1 \leq |K| < m$.
Hence, there are two indices $i < j \in K$ that are far enough,
\emph{i.e.} such that $j-i > m^{|Q|-1}$, with 
$k \notin K$ for all $i<k<j$. Therefore, for every $k$ between $i$ and
$j$, $G_k(T_{k-1}) = T_{k}$ and $G_k (U_{k-1}) = U_{k}$: no agent is transferred from $(T_{k})_{i \leq k \leq j}$ to $(U_{k})_{i \leq k \leq j}$ in the fragment $G_i \cdots G_j$. 
We say that $(T_{k})_{i \leq k \leq j}$ and $(U_k)_{i \leq k \leq j}$
do not communicate. 
In particular, we have 2 non-empty disjoint subsets $T_i,U_i$ of $Q$. We will inductively partition at each step at least one of the sequences into 2 non-communicating subsequences. 
Eventually, we obtain $|Q|+1$ 
non-communicating subsequences, and in particular 
$|Q|+1$ non-empty disjoint subsets of $Q$, a contradiction.

Applying Lemma~\ref{l.leak} again on $H'=G_i \cdots G_j$ yields a
partition  $S'= T' \uplus U'$. We define in the
same way $U'_j=U'$ and inductively $U'_k$ for $k = j-1, \ldots, i$.
As above, there are less than $m$ indices $i' \in [i,j]$ such that
$G_{i'} ( T'_{i'-1}) \cap U_{i'} \neq \emptyset$.  We can thus find
an interval $[i',j'] \subsetneq [i,j]$ such that $j' -i' > m^{|Q|-2}$
and $(T'_{k})_{i' \leq k \leq j'}$ and $(U'_k)_{i' \leq k \leq j'}$
do not communicate.

To sum up,  (any pair of) the four following sequences do not communicate together:
\begin{itemize}
 \item $(T_k \cap T'_k)_{i' \leq k \leq j'}$, $(T_k \cap U'_k)_{i' \leq k \leq j'}$,
 \item $(U_k \cap T'_k)_{i' \leq k \leq j'}$ and $(U_k \cap U'_k)_{i' \leq k \leq j'}$.
\end{itemize}

For any of these four sequences $(X_k)_{i' \leq k \leq j'}$, if 
$X_k \neq \emptyset$ for some $k$, then
$X_k \neq \emptyset$ for all $k$ (as these 4 sequences do not communicate
together and they partition $S$),
in which case we say that the sequence is {\em non-empty}.
Now some of these sequences may be empty. 
Yet, we argue that at least 3 of them are non-empty.
  
Indeed, for at least one index $i \leq k < j$, we have $G_{k}(T'_{k-1}) \cap U'_{k} \neq \emptyset$.
As there is no communication between 
$(T_k)_{i \leq k \leq j}$ and $(U_k)_{i \leq k \leq j}$,
at least one of $T_{k},U_{k}$, let say $T_{k}$, contains at least one state from $T'_{k}$ \emph{and} one state from $U'_{k}$.
Assuming $|j'-i'|$ maximal, we can choose $k=i'$ (else we have a contradiction with $|j'-i'|$ maximal, unless $i'=i$, in which case we can choose $k=j'+1$). 
That is,
$T_{i'} \cap T'_{i'} \neq \emptyset$ and $T_{i'} \cap U'_{i'} \neq \emptyset$,
and both sequences are non-empty.
Obviously at least one of $U_{i'} \cap T'_{i'}$ and $U_{i'} \cap U'_{i'}$ should be
non-empty as $U_{i'} = (U_{i'} \cap T'_{i'}) \cup (U_{i'} \cap U'_{i'})$ is nonempty, and this gives us the third non-empty sequence.
Hence, we have three non-empty sequences such that no pair of these sequences communicate between $i'$ and $j'$.

We can iterate once more to obtain $i''<j''$ with 
$j'' - i'' > m^{|Q|-3}$, and 
four non-empty sequences, such that no pair of these sequences
communicate between $i''$ and $j''$. 
This is because 1 non-empty sequence contains states from both $U''$ and $T''$, giving 2 non-empty sequences, and the 2 other non-empty sequences gives at least 2 non-empty sequences.
Obviously, this operation can be made at most $|Q|$ times, as it would
result into $|Q|+1$ non-empty and pairwise disjoint subsets of $Q$. We
thus obtain a contradiction with the number of steps being more than
$m^{|Q|} \times 2^{|Q|^4}$.
\end{proof}

\newpage


	\section{Lower bounds}

	\label{sec:lowerbounds}

The proofs of Theorems~\ref{th1} and~\ref{th2} are concluded by 
the proofs of lower bounds.

\begin{thm}
\label{th.lower}
The population control problem is \EXPTIME-hard.
\end{thm}

\begin{proof}
  We first prove \PSPACE-hardness of the population control problem,
  reducing from the halting problem for polynomial space Turing
  machines.  We then extend the result to obtain the
  \EXPTIME-hardness, by reducing from the halting problem for
  polynomial space {\em alternating} Turing machines.  Let
  $\mathcal{M} = (S,\Gamma,T,s_0,s_f)$ be a Turing machine with
  $\Gamma = \{0,1\}$ as tape alphabet. By assumption, there exists a
  polynomial $P$ such that, on initial configuration
  $x \in \{0,1\}^n$, $\mathcal{M}$ uses at most $P(n)$
  tape cells. A transition $t \in T$ is of the form
  $t=(s, s', b, b', d)$, where $s$ and $s'$ are, respectively, the
  source and the target control states, $b$ and $b'$ are,
  respectively, the symbols read from and written on the tape, and
  $d \in \{\leftarrow,\rightarrow,-\}$ indicates the move of the tape
  head. From $\mathcal{M}$ and $x$, we build an NFA
  $\nfa = (\states, \Sigma, \state_0,\Delta)$ with a distinguished
  state $\textbf{Acc}$ such that, $\mathcal{M}$ terminates in $s_f$ on
  input $x$ if and only if $(\nfa,\textbf{Acc})$ is a positive instance of
  the population control problem.

  Now we describe the states of NFA $\nfa$. They are given by:
  \begin{equation*}
  \states = \states_{\cells} \cup \states_{\pos} \cup
  \states_{\cont} \cup \{\state_0,\textbf{Acc},\frownie\}
  \end{equation*}
  where
  \begin{itemize}
  \item $\states_{\cells} = \bigcup_{i=1}^{P(n)} \{0_i,1_i\}$ are the states for
    the cells contents of $\mathcal{M}$, one per bit and per position;
  \item $\states_{\pos} = \{p_i \mid 1 \leq i \leq P(n)\}$ are the states for
    the position of tape head of $\mathcal{M}$;
  \item $\states_{\cont} = S$ are the states for the control state of $\mathcal{M}$;
  \item $q_0$ is the initial state of $\nfa$, $\textbf{Acc}$ is a sink
    winning state and $\frownie$ is a sink losing state.
  \end{itemize}
  A configuration of the Turing machine, of the form $(q, p, x) \in S
  \times [P(n)] \times \{0,1\}^{P(n)}$, 
  is represented by any configuration of $\nfa^m$
  such that the set of states with at least one agent is 
 $\{q, p\} \cup \{0_i \mid x_i = 0 \} \cup \{1_i \mid x_i =
  1\}$.

  With each transition $t=(s, s', b, b', d)$ in the Turing machine and each position $p$ of
  the tape, we associate an action $a_{t,p}$ in $\nfa$, which
  simulates the effect of transition $t$ when the head position is
  $p$.
  For instance, Fig.~\ref{TM} represents the
  transitions associated with action $a_{t,k}$,
  for the transition $t=(q_i, q_j, 0, 1, \rightarrow)$ of the Turing Machine 
  and position $k$ on the tape. Note that if agents are in the states representing
  a configuration of the Turing machine, then the only action \playerone\ can
  take to avoid $\frownie$ is to play $a_{t,p}$ where $p$ is the current head
  position and $t$ is the next allowed transition. Moreover on doing this, the
  next state in $\nfa^m$ exactly represents the next configuration of the Turing
  machine.

There are also winning actions called $\trash(Q')$ for certain subsets $Q' \subseteq Q$. 
\playerone\ should only play these when no agents are in $Q'$.
  One of them is for $Q' = \states_{\cont} \setminus \{s_f\}$ which can
  effectively only be played when the Turing machine reaches $s_f$, indicating that $\mathcal{M}$ has accepted the input $x$. $\trash(Q')$
  for other subsets $Q'$ are used to ensure that \playertwo\ sets up
  the initial configuration of the Turing machine correctly (see the formal definition of the transitions at the end of the construction). 
 

\tikzset{rect state/.style={draw,rectangle}}

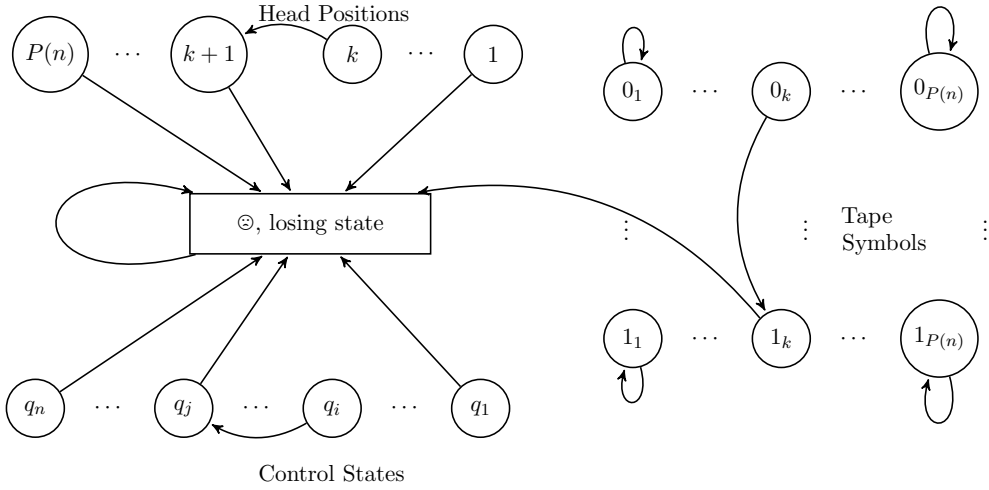
\begin{figure}[h!]
\centering
\begin{tikzpicture}[->,>=stealth',shorten >=1pt,auto,node distance=1.7cm, semithick,scale=0.8, every node/.style={scale=0.8}]
    \tikzstyle{every state}=[fill=none,draw=black,text=black]

    \node[rect state, minimum width=4cm, minimum height=1cm]         (L)         {$\frownie$, losing state};
    \node[state, below left=2.5cm of L]         (QN) []    {$q_n$};
    \node[state, right=0.2cm of QN]         (QD3) [draw=none]    {$\cdots$};
    \node[state, right=0.2cm of QD3]         (QJ) []    {$q_j$};
    \node[state, right=0.2cm of QJ]         (QD2) [draw=none]    {$\cdots$};
    \node[state, right=0.2cm of QD2]         (QI) []    {$q_i$};
    \node[state, right=0.2cm of QI]         (QD1) [draw=none]    {$\cdots$};
    \node[state, right=0.2cm of QD1]         (Q1) []    {$q_1$};
    
    \node[state, above left=2.1 cm of L]         (MN) []    {$P(n)$};
    \node[state,right=0.15cm of MN]         (HD1) [draw=none]    {$\cdots$};
    \node[state, right=0.15cm of HD1]         (MJ) []    {$k+1$};
    \node[state, right=1cm of MJ]         (MI) []    {$k$};
    \node[state, right=0.15cm of MI]         (HD3) [draw=none]    {$\cdots$};
    \node[state, right=0.15cm of HD3]         (M1) []    {$1$};
    
    \node[state, below right=0.7 cm of M1]         (A1) [right of=M1]    {$0_1$};
    \node[state, right=0.2cm of A1]         (AD1) [draw=none]    {$\cdots$};
    \node[state, right=0.2cm of AD1]         (AK) []    {$0_k$};
    \node[state, right=0.2cm of AK]         (AD2) [draw=none]    {$\cdots$};
    \node[state, right=0.2cm of AD2]         (AM) []    {$0_{P(n)}$};
    
    \node[state, right=2.2cm of L]         (D1) [draw=none]    {$\vdots$};
    \node[state, right=1.6cm of D1]         (D2) [draw=none]    {$\vdots$};
    \node[state, right=1.6cm of D2]         (D3) [draw=none]    {$\vdots$};

    \node[state, below=2.5 cm of A1]         (B1) []    {$1_1$};
    \node[state, right=0.2cm of B1]         (BD1) [draw=none]    {$\cdots$};
    \node[state, right=0.2cm of BD1]         (BK) []    {$1_k$};
    \node[state, right=0.2cm of BK]         (BD2) [draw=none]    {$\cdots$};
    \node[state, right=0.2cm of BD2]         (BM) []    {$1_{P(n)}$};
    
    \node [text width=5cm, anchor=west, right] at (-1,3.5)
    {Head Positions};
    \node [text width=5cm, anchor=west, right] at (-1,-4.15)
    {Control States};
    \node [text width=2.1cm, anchor=west, right] at (8.7,-0.12)
    {Tape\break Symbols};

    \path (QN) edge           node [below] {} (L)
        (QJ) edge           node [below] {} (L)
        (QI) edge [bend left]          node [below] {} (QJ)
        (Q1) edge           node [below] {} (L)
        
        (MN) edge           node [below] {} (L)
        (MJ) edge           node [below] {} (L)
        (MI) edge [bend right]          node [below] {} (MJ)
        (M1) edge           node [below] {} (L)

        (A1) edge [loop above]            node [below] {} (A1)
        (B1) edge [loop below]            node [below] {} (B1)
        (AM) edge [loop above]            node [below] {} (AM)
        (BM) edge [loop below]            node [below] {} (BM)
        (AK) edge [bend right]             node       {} (BK)
        (BK) edge [bend right]            node [below] {} (L)
		  (L) edge [loop left]            node [left] {} (L);
    
\end{tikzpicture}
\caption{Transitions associated with action $a_{t,k}$ for $t=(q_i, q_j, 0, 1, \rightarrow)$.}
\label{TM}
\end{figure}
\medskip

Let us now describe the transitions of $\nfa$ in more detail. The actions are
\begin{itemize}
\item $\Sigma = \Sigma_{\transition} \cup \Sigma_{\trash} \cup \{\start\}$, with
\begin{itemize}
\item $\Sigma_{\transition} = T \times \left\{ 1, \ldots, P(n) \right\}$
\item $\Sigma_{\trash} = \Bigl\{\trash(\states') \mid \states' \in
 \bigl\{\{0_i,1_i\}_{1 \leq i \leq P(n)}, \states_{\cont}, \states_{\cont}\setminus\{s_f\}, \states_{\pos}\bigr\} \Bigr\}$
\end{itemize}
\end{itemize}


  To describe the effect of actions from $\Sigma$, we use the
  following terminology: a state $\state \in \states$ is called an
  $\alpha$-sink for $\alpha \in \Sigma$ if $\Delta(\state,\alpha) =
  \{\state\}$, and it is a sink if it is an $\alpha$-sink for every
  $\alpha$. Only the initial $\start$ action is nondeterministic: for
  every $\alpha \in \Sigma \setminus\{\start\}$, and every $\state \in
  \states$, $\Delta(\state,\alpha)$ is a singleton. A consequence is
  that in the games $\nfa^m$, the only decision Player~2 makes is in
  the first step.

  The two distinguished states $\textbf{Acc}$ and $\frownie$ are
  sinks. Moreover, any state but $q_0$ is a $\start$-sink. From the
  initial state $\state_0$, the effect of $\start$ aims at
  representing the initial configuration of the Turing machine:
  $\Delta(\state_0,\start) = \{s_0,p_1\} \cup \{0_k \mid x_k =0
  \textrm{ or } k \geq n\} \cup \{1_k \mid x_k =1\}$. Then, the
  actions from $\Sigma_\transition$ simulate transitions of the Turing
  machine. Precisely, the effect of action $\alpha = ((s_s, s_t, b_r, b_w, d), i) \in
    \Sigma_\transition$ is deterministic and as follows:
\begin{itemize}
\item $\Delta(s_s,\alpha) = \{s_t\}$
\item $\Delta(p_i,\alpha) = \begin{cases}\{p_{i+1}\} & \textrm{ if } i<P(n)
    \textrm{ and } d=\rightarrow\\ \{p_{i-1}\} & \textrm{ if } i>1
    \textrm{ and } d=\leftarrow\\ \{p_i\} & \textrm{ otherwise;}
\end{cases}$
\item $\Delta(0_i,\alpha) = \begin{cases} \{0_i\} & \textrm{ if } b_r =0
    \textrm{ and } b_w=0\\
\{1_i\} & \textrm{ if } b_r =0
    \textrm{ and } b_w=1
\end{cases} \quad \Delta(1_i,\alpha) = \begin{cases} \{0_i\} & \textrm{ if } b_r =1
    \textrm{ and } b_w=0\\
\{1_i\} & \textrm{ if } b_r =1
    \textrm{ and } b_w=1;
\end{cases}
$
\item $\Delta(0_j,\alpha) = 0_j$ and $\Delta(1_j,\alpha) = 1_j$, for
  $j \neq i$;
\item otherwise, $\Delta(q,\alpha) = \frownie$.
\end{itemize}
Last, we described how actions from $\Sigma_\trash$ let the system
evolve. Let $\trash(Q') \in \Sigma_\trash$ be a check action for set
$Q' \subseteq Q$. Then  
\begin{itemize}
\item $\Delta(q,\trash(Q')) = \begin{cases} \frownie & \textrm{ if }
    q \in Q' \cup \{q_0,\frownie\}\\ \textbf{Acc} &\textrm{ otherwise.}\end{cases}$
\end{itemize}

We claim that this construction ensures the following equivalence:

\begin{lem}
\label{lem:eq-MT-SS}
$\mathcal{M}$ halts on input $x$ in $q_f$ if and only if $(\nfa,\textbf{Acc})$ is
a positive instance of the sure-synchronization problem.
\end{lem}

\begin{proof}\leavevmode
\begin{itemize}
\item case $m \leq P(n)+1$: not enough tokens for player 2 in the first step to cover all of $\Delta(q_0,\start)$; player 1 wins in the next step
  by selecting the adequate check action
\item case $m \geq P(n)+2$: best move for player 2 in the first step is
  to cover all of $\Delta(q_0,\start)$; afterwards, if player 1 does
  not mimic the execution of the Turing machine, some tokens get
  stuck in $\frownie$; thus the best strategy for player 1 is to
  mimic the execution of the Turing machine; then, the machine halts
  if and only if all the tokens in $\states_{\cont}$ converge to $s_f$. Now applying $\trash(\states_\cont  \setminus \{s_f\})$ moves all tokens to $\textbf{Acc}$.\qedhere
\end{itemize}
\end{proof}

We thus performed a \PTIME reduction of the halting problem for
polynomial space Turing machines to the sure synchronization problem,
which is therefore \PSPACE-hard.

\medskip

Now, in order to encode an alternating Turing machine, we assume that
the control states of $\mathcal{M}$ alternate between states of \playerone\
and states of \playertwo. The NFA $\nfa$ is extended with a
state $\smiley$, 
a state $C$, which represents that \playerone\ decides what transition to take, and one state $q_t$ per transition $t$
of $\mathcal{M}$, which will represent that \playertwo\ chooses to play transition $t$ as the next action.  Assume first, that $C$ contains at
most an agent; we will later explain how to impose this.
 
The NFA has an additional transition labelled
$\init$ from $q_0$ to $C$, and 
one transition from $C$ to every state $q_t$
labeled by $a_{t',p}$,
for every transition $t$ and action $a_{t',p}$.
Intuitively, 
whatever action $a_{t',p}$ is played by \playerone,
\playertwo\ can choose 
the next action to be associated with $t$ by 
placing the agent to state $q_t$.

From state $q_t$, only actions of the form $a_{t,p}$ are allowed, leading
back to $C$. That is, actions $a_{t',p}$ with $t' \neq t$ lead from
$q_t$ to the sink losing state $\frownie$. This encodes that \playerone\
must follow the transition $t$ chosen by \playertwo. To punish
\playertwo\ in case the current tape contents is not the one expected by
the transition $t=(s, s', b, b', d)$ he chooses, there are checking actions
$\trash_{s}$ and $\trash_{p,b}$ enabled from state $q_t$.  Action
$\trash_{s}$ leads from $q_t$ to $\smiley$, and also from $s$ to
$\frownie$. Similarly,
$\trash_{p,b}$ for any position $p$ and $b \in \{0, 1\}$ leads from $q_t$ to $\smiley$ and from any position state
$q \neq p$ to $\frownie$, and from $b_p$ to $\frownie$. 
In this way, \playertwo\ will not move the token from $C$ to 
an undesired $q_t$.
 This
ensures that \playertwo\ places the agents only in a state $q_t$ which
agrees with the configuration.

Last, there are transitions on action $\ttend$ from state $\smiley$, $C$ and any of the $q_t$'s to the target state $\smiley$.
Action $\ttend$ loops around the accepting state $\textbf{Acc}$ associated with the Turing machine, and 
it leads from any other state to $\frownie$.  
Last, there is an action $\winact$, which leads from $\smiley$ to $\smiley$, from $\textbf{Acc}$ to $\smiley$,
and from any other state to $\frownie$.
This action $\winact$ may seem unnecessary, 
but its purpose will appear clear in the following step.
This whole construction
encodes, assuming that there is a single agent in $C$ after the first
transition, that \playerone\ can choose the transition from a
\playerone\ state of $\mathcal{M}$, and \playertwo\ can choose
the transition from an \playertwo\ state.

\medskip

 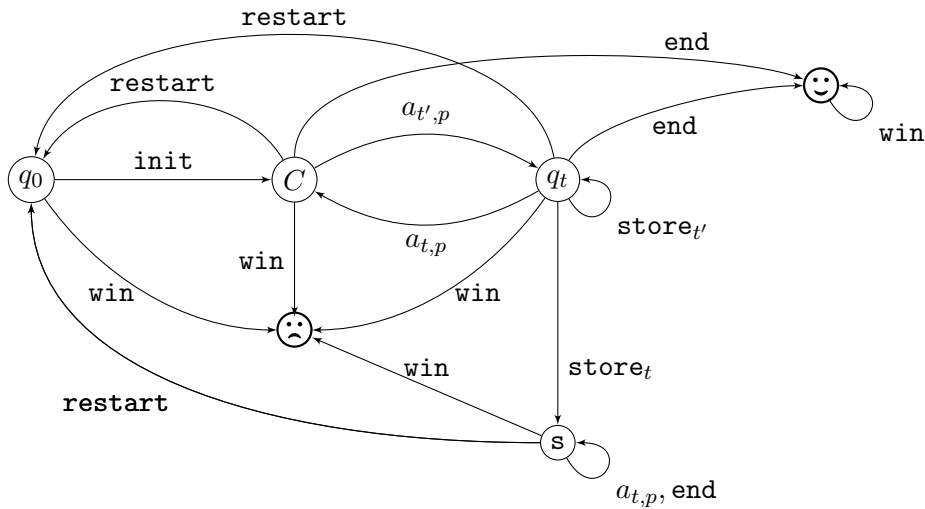
\begin{figure}[b!]
\centering
\begin{tikzpicture}[scale=1]
\draw (-3.5,0) node [circle,draw,inner sep=2pt,minimum size=12pt ]
  (s0) {$\state_0$} ;

  \draw(0,0) node [circle,draw,inner sep=2pt,minimum size=12pt ]
  (c) {$C$} ;

  \draw(3.5,0) node [circle,draw,inner sep=2pt,minimum size=12pt ]
  (t) {$q_t$} ;

\draw(3.5,-3.5) node [circle,draw,inner sep=2pt,minimum size=12pt ]
  (store) {$\sstore$} ;

\draw(7,1.25) node[inner sep=-2pt] (win) {$\huge{\smiley}$} ;

\draw(0,-2) node [inner sep=-2pt]
  (lose) {$\huge{\frownie}$} ;

\draw [-latex'] (c) -- (lose) node[pos=.5,left] {$\winact$};
\draw [-latex'] (s0) .. controls +(305:1.5cm) and +(180:1.5cm) .. (lose) node[pos=.5,left] {$\winact$};
\draw [-latex'] (t) .. controls +(235:1.5cm) and +(0:1.5cm) .. (lose) node[pos=.5,right] {$\winact$};
\draw [-latex'] (store) -- (lose) node[pos=.5,above] {$\winact$};

 \draw [-latex'] (s0) -- (c) node[pos=.5,above] {$\init$};
 \draw [-latex'] (c) .. controls +(30:1.5cm) and +(150:1.5cm)  .. (t)
 node [pos=.5,above] {$a_{t',p}$};
\draw [-latex'] (t) .. controls +(210:1.5cm) and +(330:1.5cm)  .. (c)
 node [pos=.5,below] {$a_{t,p}$};

\draw [-latex'] (store) .. controls +(180:4cm) and +(270:2.5cm)  .. (s0)
 node [pos=.5,below left] {$\restart$};

\draw [-latex'] (c) .. controls +(120:1.5cm) and +(60:1.5cm)  .. (s0)
 node [pos=.5,above] {$\restart$};

\draw [-latex'] (c) .. controls +(90:2cm) and +(160:1.5cm)  .. (win)
 node [pos=.75,above] {$\ttend$};

 \draw [-latex'] (t) .. controls +(60:1cm) and +(180:1cm) .. (win) node [pos=.5,below] {$\ttend$};

\draw [-latex'] (t) .. controls +(100:2.5cm) and +(80:2.5cm)  .. (s0)
 node [pos=.5,above] {$\restart$};

\draw [-latex'] (store) .. controls +(180:4cm) and +(270:2.5cm)  .. (s0)
 node [pos=.5,below left] {$\restart$};
\draw [-latex'] (t) -- (store) node[pos=.75,right] {$\store_t$};

\draw [-latex']  (t) .. controls +(300:30pt) and +(0:30pt) .. (t)
 node[midway,below right]{$\store_{t'}$};

\draw [-latex']  (win) .. controls +(300:30pt) and +(0:30pt) .. (win)
 node[midway,below right]{$\winact$};

\draw [-latex']  (store) .. controls +(300:30pt) and +(0:30pt) .. (store)  node[midway,below right]{$a_{t,p},\ttend$};

\end{tikzpicture}
\caption{Gadget simulating a single agent in $C$.}
\label{fig:gadget-unique}
\end{figure}


Let us now explain 
how to deal with the
case where \playertwo\ places several agents in state $C$ on the
initial action $\init$, enabling the possibility to later send agents
to several $q_{t}$s simultaneously.
With the current gadget, 
if there is an agent in $q_{t1}$ and one in $q_{t2}$,
then \playerone\ would be stuck as playing $a_{t1,k}$ would send 
the agent from $q_{t2}$ to $\frownie$, and vice-versa.
To handle this case, consider the gadget from Figure~\ref{fig:gadget-unique}.
We use an extra 
state $\sstore$, actions $\store_t$ for each transition $t$, and action
$\restart$.

Action $\store_{t}$  leads from $q_t$ to $\sstore$, and loops
on every other state.
From all states except $\smiley$ and $\frownie$
(in particular, $\sstore$ and every state associated with the Turing machine, including $\textbf{Acc}$), action $\restart$
leads to $q_0$. Last, the effects of $\ttend$ and $\winact$ are
extended as follow to $\sstore$: 
$\ttend$ loops on $\sstore$, while 
$\winact$ leads from $\sstore$ to $\frownie$.

Assume that input $x$ is not accepted by the alternating Turing
machine $\mathcal{M}$, and let $m$ be at least $P(n)+3$.
In the $m$-population game, \playertwo\ has a winning strategy placing
initially a single agent in state $C$. If \playerone\ plays
$\store_{t}$ (for some $t$), either no agents are stored, or the
unique agent in $C$ is moved to $\sstore$. 
Playing $\ttend$ does not change the configuration, and 
\playerone\ cannot
play $\winact$. Thus, there is no way to lead the agents encoding the Turing machine configuration to $\smiley$.
Playing $\restart$ moves all the agents back to the original configuration $q_0$. This shows that $\store_{t}$ is
useless to \playerone\, and thus \playertwo\ wins as in the previous case.

Conversely, assume that \playerone\ has a strategy in $\mathcal{M}$ witnessing the acceptance of $x$.
If \playertwo\ never split, then \playerone\ never plays any store actions and wins as in the previous case.
Otherwise, assume that \playertwo\ places at least two agents in $C$ to eventually split them to ${t_1}, \ldots, {t_n}$. 
In this case, \playerone\ can play the corresponding actions
$\store_{t_2}, \ldots, \store_{t_n}$ moving all agents (but the ones
in ${t_1}$) in $\sstore$, after which he plays his winning strategy
from ${t_1}$ resulting in sending at least one agent to $\smiley$. Then, \playerone\ plays $\restart$ and proceeds inductively with strictly less agents from $q_0$, until there is no agent in $C,q_t$
in which case \playerone\ plays $\winact$ to win.
\end{proof}

Surprisingly, the cut-off can be as high as doubly exponential in the size of the NFA.

\begin{prop}
\label{prop:cutoff-lowerbound}
  There exists a family of NFA $(\nfa_n)_{n \in \nats}$ such that
  $|\nfa_n|=2n+7$, 
  and for $M = 2^{2^{n}+1}+n$, there is no winning strategy
  in $\nfa_n^{M}$ and there is one in $\nfa_n^{M-1}$.
\end{prop}

{\begin{proof}
  Let $n\in \nats$. The NFA $\nfa_n$ we build is 
  the disjoint union of {\em two} NFAs with different properties, called
 $\splitnfa$ and $\countnfa{n}$.
 On the one hand, for $\splitnfa$, it requires $\Theta(\log m)$ steps for
 \playerone\ to win the $m$-population game. On the other
  hand, $\countnfa{n}$ implements a
  usual counter over $n$ bits,
  such that \playerone\ can avoid losing for $O(2^n)$ steps. In the combined NFA $\nfa_n$, we require that \playerone\ win in $\splitnfa$
    \emph{and} avoid losing in $\countnfa{n}$. This ensures that  $\nfa_n$ has a cutoff of $\Theta(2^{2^n})$

 Recall Figure~\ref{fig:splitgadget}, which presents the splitting gadget
 $\splitnfa$. It has the following properties. In $\splitnfa^m$ with $m \in \nats$ agents:
\begin{description}
\item[$(s1)$ \playerone\ has a strategy to ensure win in $2 \left \lfloor \log_2 m \right \rfloor + 2$ steps] Consider the following strategy of \playerone.
  \begin{itemize}
  \item Play $\delta$ if there is at least one agent in state $q_0$. Otherwise,
  \item Play $a$ if the number of agents in $q_1$ is greater than in $q_2$.
  \item Play $b$ if the number of agents in $q_2$ is greater than or same as in $q_1$.
  \end{itemize}
  
	For this strategy, let us look at the number of agents in the state $q_0$ and $f$
  respectively. For instance, since the play starts from all agents in $q_0$, the starting state has count $(m,0)$ - there are $m$ tokens in $q_0$ and $0$ in $f$. From
  a state with counts $(k,m-k)$, after two steps of this strategy (regardless of
  \playertwo's play), we will end up in a state with counts $(l,m-l)$ for some
  $l \leq \left \lfloor k/2 \right \rfloor$. Hence within $2\times\lfloor \log_2
  m\rfloor + 2$ steps starting from the initial state, one will reach a state
  with counts $(0,m)$ and the \playerone\ wins.

\item[$(s2)$ No strategy of \playerone\ can ensure a win in less than $2 \left
    \lfloor \log_2 m \right \rfloor + 2$ steps] The transition from $q_0$ on 
  $\delta$ is the only real choice \playertwo\ has to make. Assume that \playertwo\
  decides to send an equal number of agents (up to a difference of 1) to both $q_1$
  and $q_2$ from $q_0$. Against this strategy of \playertwo, let $\alpha_1, \alpha_2, \ldots \alpha_k$ be the
  shortest sequence of actions by \playerone\ which lead all the agents
  into $f$. 
  
  We now show that $k \geq 2\times \lfloor \log_2 m \rfloor + 2$.
  Initially all agents are in $q_0$. So by the minimality of $k$ we should
  have $\alpha_1 = \delta, \alpha_2 \in \{a, b\}, \alpha_3 = \delta, \ldots
  \alpha_k \in \{a, b\}$, since other actions will not change the state of the
  agents. For any $i \in \{ 1, 2, \ldots \frac{k}{2}\}$,
  after \playerone\ plays $\alpha_{2i}$, denote the number of agents in $q_0$ and $f$
  by $(l_i, m - l_i)$. Note $l_{i} \geq \frac{l_{i-1} - 1}{2}$, since
  \playertwo\ sends equal number of agents from $q_0$ to $q_1$ and $q_2$.
  Iterating this equation for any $j \in \{ 1, 2, \ldots \frac{k}{2}\}$ gives $l_j \geq \frac{l_0+1}{2^j} - 1$,
  where $l_0 = m$ is the number of agents in state $q_0$ at the beginning. In
  particular this shows that $\frac{k}{2} > \log m$ since $l_{\frac{k}{2}} = 0$.
  Hence $k \geq 2 \times (\lfloor  \log_{2}m \rfloor + 1)$.

\end{description}


	\begin{figure}[t!]
		\centering
		\begin{tikzpicture}[>=latex',->,scale=0.8, every node/.style={scale=0.8}]
		
		\node[initial above, initial text=,state] (q0) {$q_0$};

		\node[state,below=of q0, yshift=-1.5cm]  (li) {$l_i$};
		
		
		\node[state, right=of li, xshift=2cm] (l1) {$l_1$};
		\node[state, left=of li, xshift=-2cm] (ln) {$l_n$};
		
		\node[state,below=of li, distance=2cm]  (hi) {$h_i$};
		\node[state, below=of ln] (hn) {$h_n$};
		\node[state, below=of l1] (h1) {$h_1$};
		
		\node[below left=of li, yshift=0.4cm] {$\cdots$};
		\node[below right=of li, yshift=0.4cm] {$\cdots$};
				
		\node[state, below=of hi, inner sep=0, yshift=-1.5	cm] (sad) {$\hFrownie$};

		\path (li) edge node[left,pos=0.5] {$\alpha_i$} (hi);
		\path (hi) edge node[right,pos=0.5] {$\alpha_i$} (sad)
			  (li) edge[in=120,out=220] node[left, pos=0.5] {$\{\alpha_j\}_{j>i}$} (sad)
			  (hi) edge[in=-20, out=30] node[right,pos=0.7] {$\{\alpha_j\}_{j>i}$} (li)
			  (hi) edge[in=-30, out=-70,looseness=5] node[right=1pt] {$\{\alpha_j\}_{j<i}$} (hi) 
			  (li) edge[in=150,out=110, looseness=6] node[left] {$\{\alpha_j\}_{j<i}$}(li);
			  
		\path (q0) edge node[left]	 {$\alpha_1, \ldots$} node[right]	 {$\ldots ,\alpha_n$} (li);
		
		\path (q0) edge node {} (l1)
			  (q0) edge  (ln);
		
        \path (l1) edge[bend left] node[pos=0.7, right] {$\alpha_1$} (h1)
			  (ln) edge node[right] {$\alpha_n$} (hn)
			  (h1) edge node[left] {$\alpha_1$} (sad)
			  (hn) edge node[left] {$\alpha_n$} (sad);
        \path
        (sad) edge[loop below] node {$\alpha_1,..,\alpha_n$} (sad)
        (h1)  edge node[left] {$\{\alpha_j\}_{j>1}$} (l1)
        (hn) edge[looseness=5, in=190, out=-120] node[left] {$\{\alpha_j\}_{j<n}$} (hn)
        (ln) edge[looseness=5, in=160, out=110 ] node[left] {$\{\alpha_j\}_{j<n}$} (ln)
        ;

		\draw[->] (l1) .. controls ($(l1)+(2cm,0cm)$) and ($(h1)+(2cm,-2cm)$) .. node[right, pos=0.7] {$\{\alpha_j\}_{j>1}$}(sad);		
		   
		\end{tikzpicture}
		\caption{The counting gadget $\countnfa{n}$.}	
                \label{fig:countgadget}
	\end{figure}

  \medskip
  
The gadget $\countnfa{n}$, shown in Figure~\ref{fig:countgadget}, represents a
binary counter. For each $i \in \{ 1 \ldots n \}$, it has states $\ell_i$
(meaning bit $i$ is $0$) and $h_i$ (meaning bit $i$ is $1$) and actions
$\alpha_i$ (which sets the $i$th bit and resets the bits for $j < i$). The only
real choice that \playertwo\ has in this gadget is at the first step, and it is
optimal for \playertwo\ to place agents in as many states as possible -- in particular when $m \geq n$, placing one agent in each $l_i$ for $i \in \{ 1 \ldots n \}$.

After this, the play is completely determined by \playerone's actions. Actually,
\playerone\ doesn't have much choice when there is at least one agent is each $l_i$. \playerone\ must simulate an $n$-bit counter if it wants to prevent some agent from reaching $\frownie$. More
precisely, assume inductively that all the agents are in the states given by
$b_n b_{n-1} \ldots b_1$ where $b_k \in \{ l_k , h_k \}$ (initially $b_i=l_i$
for each $i$). If $ b_i = h_i$ for each $i$, then any action $\alpha_i$ will lead
the agent in state $h_i$ to $\frownie$. Otherwise, let $j$ be the smallest index
such that $b_j = l_j$. Observe that $\alpha_j$ is the only action that doesn't lead some agent to $\frownie$; $\alpha_i$ for $i < j$ would lead the agent in $h_i$ to
$\frownie$, while $\alpha_i$ for $i > j$ would lead the agent in $l_j$ to
$\frownie$. On playing $\alpha_j$, the agents now move to the states $b_n
b_{n-1} \ldots b_{j-1} h_j l_{j-1} \ldots l_0$ -- which can be interpreted as
the next number in a binary counter.

This means that the gadget $\countnfa{n}$ has the following properties: 
\begin{itemize}
\item[$(c1)$] For any $m$, \playerone\ has a strategy in the $m$-population game on 
$\countnfa{n}$ to avoid $\frownie$ for $2^n$ steps by playing $\alpha_i$ whenever the counter suffix from bit $i$ is $0 1\cdots 1$;
\item[$(c2)$] For $m \geq n$, no strategy of \playerone\ in $\countnfa{n}^m$ can avoid
$\frownie$ for more than $2^n$ steps.
\end{itemize}

To construct $\nfa_n$, the two gadgets $\splitnfa$ and $\countnfa{n}$ are combined
by adding a new initial state, and an action labeled {\em init} leading from
this new initial state to the initial states of both NFAs.
Actions for $\nfa_n$ are made up of pairs of actions, one for each gadget: $ \{a,b,\delta\} \times \{\alpha_i \mid 1 \leq i \leq n\}$.
We further add an action $*$ which can be played from any state of $\countnfa{n}$ except $\frownie$, and only from $f$ in $\splitnfa$, 
leading to the global target state $\smiley$.

Let $M=2^{2^{n}+1}+n$. 
We deduce that the cut-off is $M-1$ as follows:
\begin{itemize}
\item 
For $M$ agents, a winning strategy for \playertwo\ is to
first split $n$ tokens from the initial state to the $q_0$ 
of $\countnfa{n}$, in order to fill each $l_i$ with 1 token, and
$2^{2^{n}+1}$ tokens to the $q_0$ of $\splitnfa$.
Then \playertwo\ splits evenly tokens between $q_1,q_2$ in $\splitnfa$.
In this way, \playerone\ needs at least $2^n+1$ steps
to reach the final state of $\splitnfa$ ($s2$), but \playerone\ reaches $\frownie$ after these $2^n+1$ steps in $\countnfa{n}$ ($c2$).

\item
For $M-1$ agents, \playertwo\ needs to use at least $n$ tokens 
from the initial state to the $q_0$ of $\countnfa{n}$, 
else \playerone\ can win easily. But then there are less than 
$2^{2^{n}+1}$ tokens in the $q_0$ of $\splitnfa$.
And thus by $(s1)$, \playerone\ can reach $f$ within $2^n$ steps,
after which he still avoids $\frownie$ in $\countnfa{n}$ ($c1$).
And then \playerone\ sends all agents to $\smiley$ using $*$.
\end{itemize}
Thus, the family $(\nfa_n)$ of NFA exhibits a doubly exponential cut-off.
\end{proof}
}

\newpage


\section{Discussion}
\label{sec:conclu}
Obtaining an \EXPTIME algorithm for the control problem of a population of agents was challenging. We also managed to prove a matching lower-bound. Further, the surprising doubly exponential matching upper and lower bounds on the cut-off imply that the alternative technique, checking that \playerone\ wins all $m$-population game for $m$ up to the cut-off, is far from being efficient. 

The idealised formalism we describe in this paper is not entirely satisfactory: for instance, while each agent can move in a non-deterministic way, unrealistic behaviours can happen, \emph{e.g.}\ all agents synchronously taking infinitely often the same choice.  An almost-sure control problem in a probabilistic formalism should be studied, ruling out such extreme behaviours.
As the population is discrete, we may avoid the undecidability that holds for distributions~\cite{DMS12} and is inherited from the equivalence with probabilistic automata~\cite{GO10}. Abstracting continuous distributions by a discrete population of arbitrary size could thus be seen as an approximation technique for undecidable formalisms such as probabilistic automata.

\medskip

{\bf Acknowledgement:} We are grateful to Gregory Batt for fruitful discussions concerning the biological setting. Thanks to Mahsa Shirmohammadi for interesting discussions. This work was partially supported by ANR project STOCH-MC (ANR-13-BS02-0011-01), and by DST/CEFIPRA/Inria Associated team EQUAVE.


\bibliographystyle{plainurl}
\bibliography{biblio}

\end{document}